\documentclass[11pt]{amsart}
\usepackage[T1]{fontenc}
\usepackage[utf8]{inputenc}
\usepackage{lmodern}
\usepackage{amsmath,amssymb,amsthm,mathtools}
\usepackage[margin=1in]{geometry}
\usepackage{microtype}
\usepackage{enumitem}
\usepackage{hyperref}
\usepackage{comment}
\usepackage{amsthm}   

\hypersetup{colorlinks=true,linkcolor=blue,citecolor=blue,urlcolor=blue}

\newcommand{\pa}{\partial}
\newcommand{\del}{\delta}

\newcommand{\mn}{\mu}

\newtheorem{theorem}{Theorem}[section]
\newtheorem{proposition}[theorem]{Proposition}
\newtheorem{lemma}[theorem]{Lemma}
\newtheorem{corollary}[theorem]{Corollary}
\theoremstyle{remark}
\newtheorem{remark}[theorem]{Remark}
\newtheorem{definition}[theorem]{Definition}

\begin{document}

\title[Gibbons--Tsarev type systems and Eventual identities]{Gibbons--Tsarev type systems and Eventual identities}
\author{Alessandro Arsie}
\address{A.~Arsie:\newline Department of Mathematics and Statistics, The University of Toledo,\newline 2801 W. Bancroft St., 43606 Toledo, OH, USA}
\email{alessandro.arsie@utoledo.edu}
\author{Paolo Lorenzoni}
\address{P.~Lorenzoni:\newline Dipartimento di Matematica e Applicazioni, Universit\`a di Milano-Bicocca, \newline
Via Roberto Cozzi 55, I-20125 Milano, Italy and INFN sezione di Milano-Bicocca}
\email{paolo.lorenzoni@unimib.it}
\author{Sara Perletti}
\address{S.~Perletti:\newline Scuola Internazionale Superiore di Studi Avanzati (SISSA), \newline
Via Bonomea 265, 34136 Trieste, Italy and INFN sezione di Trieste}
\email{sperlett@sissa.it}
\author{Karoline van Gemst}
\address{K.~van Gemst:\newline Dipartimento di Matematica e Applicazioni, Universit\`a di Milano-Bicocca, \newline
Via Roberto Cozzi 55, I-20125 Milano, Italy and INFN sezione di Milano-Bicocca}
\email{karoline.vangemst@unimib.it}
\date{}

\begin{abstract}
 We show that non-diagonalisable reductions of the dKP equation associated with regular non-semisimple $F$-manifolds cannot exist. The proof is based on the derivation and study of a generalised Gibbons--Tsarev system (gGT system) in the non-semisimple/non-diagonalisable setting. Remarkably, a class of solutions of the gGT system is defined by eventual identities of the underlying regular $F$-manifold structure. Furthermore, we use these vector fields to construct integrable reductions of Pavlov's  hydrodynamic chain. In this case, the  corresponding solutions are defined for any choice of Jordan block structure of the operator of multiplication by an eventual  identity.
\end{abstract}
\maketitle
\tableofcontents

\section{Introduction}
The existence of large families of reductions of the dispersionless KP hierarchy, or equivalently of the Benney momentum chain, governed by integrable diagonal systems of hydrodynamic type (or semihamiltonian systems \cite{ts91}) was discovered by Gibbons and Tsarev in their seminal paper \cite{GT1}. These \emph{diagonal reductions} play a central role in the modern theory of integrable dispersionless equations. In particular, the Gibbons--Tsarev mechanism became one of the fundamental tools in the programme, developed for instance in \cite{FK}, of characterising $(2+1)$-dimensional dispersionless integrable PDEs and integrable hydrodynamic chains by means of their hydrodynamic reductions.

In the Gibbons--Tsarev setting, an $n$-component diagonal reduction is described by functions
\[
\mu^1(r^1,\dots,r^n),\dots,\mu^n(r^1,\dots,r^n), \qquad
V(r^1,\dots,r^n),
\]
satisfying the system
\begin{subequations}\label{GT}
\begin{align}
\partial_i\mu^j&=\frac{\partial_i V}{\mu^i-\mu^j}, \label{gt1}\\
\partial_i\partial_jV&=\frac{2\partial_iV\partial_j V}{(\mu^i-\mu^j)^2}, \label{gt2}
\end{align}
\end{subequations}
for $i\neq j$. This is the classical \emph{Gibbons--Tsarev system}. Its solutions determine the characteristic velocities of the reduction together with a density of conservation law $V$. The same function $V$ is also the Egorov potential of the diagonal metric defining a Hamiltonian structure of the reduction; the existence of such Hamiltonian structures for generic reductions was established in \cite{GLR}. Since the system \eqref{GT} is Darboux compatible, its solutions depend on $n$ arbitrary functions of one variable.

A second remarkable discovery of Gibbons and Tsarev is that diagonal reductions admit an equivalent description in terms of Landau--Ginzburg superpotentials. These arise from a system of $n$ chordal Loewner equations, whose compatibility is exactly the Gibbons--Tsarev system, governing families of univalent conformal maps from the upper half-plane to slit domains in the upper half-plane \cite{GT2}. This picture links the hydrodynamic reduction method with complex-analytic and Frobenius-type structures.

The purpose of the present paper is to investigate the non-diagonalisable case. In several special examples it is easy to see that genuinely non-diagonalisable reductions of dKP do not occur, but, to the best of our knowledge, a general argument in the regular non-semisimple setting has not been written down explicitly. Our strategy is to translate the problem into the geometry of $F$-manifolds.

Indeed, reductions of dKP are naturally related to special semisimple $F$-manifolds with compatible connection \cite{LPR,LPVG1}. From this point of view, the semisimple case is precisely the diagonal one. It is therefore natural to ask whether regular non-semisimple $F$-manifolds may produce non-diagonalisable reductions. Starting from this geometric ansatz, one is led to a system generalising \eqref{GT}, namely the \emph{generalised Gibbons--Tsarev system}:
\begin{subequations}\label{GTgen}
\begin{align}
(\mathcal{L}_{\mu}c^i_{hj}-c^i_{js}c^s_{ht}[e,\mu]^t)\mu^h&=e(V)\delta^i_j-e^i\partial_jV, \label{gtg1}\\
(c^s_{ih}\partial_s\partial_j V-c^s_{jh}\partial_s\partial_i V)\mu^h&=
(c^h_{sj}\partial_i\mu^s-c^h_{si}\partial_j\mu^s)\partial_hV, \label{gtg2}
\end{align}
\end{subequations}
where $c^i_{jk}$ are the structure functions of the product, and $e^i$ are the components of the unit field of the $F$-manifold. In the semisimple case \eqref{GTgen} reduces to the standard Gibbons--Tsarev system.

It turns out that, in the regular non-semisimple case, \eqref{GTgen} is too rigid to admit non-trivial solutions of the type relevant for dKP. More precisely, let the local regular structure be described by Jordan blocks of sizes $m_1,\dots,m_r$. We prove that any solution necessarily satisfies
\[
\partial_{p(\alpha)}V=0, \qquad \text{for every } p\ge 2,
\]
inside each block. Hence $V$ depends only on the leading variable of each block. Once all nilpotent directions are eliminated, the surviving equations are exactly the classical Gibbons--Tsarev system on the first variables. In this precise sense there are no genuinely regular non-diagonalisable reductions of dKP.

Remarkably, the generalised Gibbons--Tsarev system admits a class of solutions related to the  geometry of the underlying $F$-manifold. Indeed, when $V$ is constant, we obtain solutions $(\mu^1,...,\mu^n)$ satisfying the condition
 \begin{equation}\label{EvId}
 \mathcal{L}_{\mu}c^i_{hj}-c^i_{js}c^s_{ht}[e,\mu]^t=0.
 \end{equation}
 In the theory of $F$-manifolds, vector fields satisfying this condition are called \emph{eventual identities}; they play an important role in the duality theory of $F$-manifolds. Thus, although the non-semisimple regular setting does not produce new dKP reductions, it still singles out a geometrically meaningful class of solutions of the generalised Gibbons--Tsarev equations.

In the final part of the paper we study a variant of the system \eqref{GTgen}  where
 the equations \eqref{gtg1} are replaced by \eqref{EvId}. We show that any  solution
   $(\mu^1,...,\mu^n,V)$ of the  new system defines a reduction of the form
   \[C^\alpha=C^\alpha(r^1,...,r^n),\qquad r^i_t=c^i_{jk}(\mu^j-Ve^j)r^k_x\]
of \emph{Pavlov's hydrodynamic chain} (see \cite{pavlov})
\begin{equation*}
C^{\alpha-1}_t=C^{\alpha}_x-C^0C^{\alpha-1}_x,\qquad \alpha\in\mathbb{Z},
\end{equation*}
for any choice of product structure. Integrable systems
 of hydrodynamic type of this form were introduced and studied in \cite{LM} from  the viewpoint of Nijenhuis geometry and in \cite{LPVG1}  from the viewpoint of the theory
  of  $F$-manifolds. The proof relies  on the fact that the operator of multiplication by an eventual identity  has vanishing  Nijenhuis torsion (see \cite{ALimrn}). Using this fact it is possible to define recursively all functions 
\[C^k(r^1,...,r^n),\qquad k=\pm 1, \pm 2, \dots,\]
starting from $C^0=V(r^1,...,r^n)$.

The paper is organised as follows. In Section 2.1 we recall basic definitions about regular $F$-manifolds and their canonical coordinates. In Section 2.2 we recall the definition of an eventual identity $\mathcal{E}$ and prove two alternative characterisations in the regular case: one in terms of a connection $\nabla$ uniquely defined by the eventual identity (Proposition \ref{NablaEvId}) and one in terms of the vanishing of Nijenhuis torsion of $\mathcal{E}\circ$  (Proposition \ref{NiEvId}). In Section 2.3, following \cite{LM}, we briefly introduce a class of integrable PDEs associated to a tensor field with vanishing Nijenhuis torsion or, in the context of regular $F$-manifolds, to eventual identities (see \cite{LPVG1}). In Section 3 we derive a system of PDEs (Proposition \ref{Prop_dKP_red}), generalising  the usual Gibbons--Tsarev system, whose solutions describe regular reductions of dKP. Section  4 is devoted to the analysis of such a  system. In particular, we prove that non-trivial non-diagonalisable reductions of dKP cannot exist, first in the case of a single Jordan block of arbitrary size (Section 4.1, Proposition \ref{non_ex_one_block}) and then in the general case (Section 4.2, Proposition \ref{non_ex_gen}). In Section 5 we show that reductions of Pavlov's hydrodynamic chain coincide  
 with the systems of PDEs associated to a tensor field $L$ with vanishing Nijenhuis torsion  described in Section 2.3 (Proposition \ref{Th_red_PHC}). In particular, due to Proposition \ref{NiEvId}, the special case where $L=\mathcal{E}\circ$ corresponds to eventual identities. Unlike the case of dKP, in the case of reductions of Pavlov's hydrodynamic chain eventual identities define non-trivial reductions for any choice of the Jordan block structure of $L$. In Section 5.1, we give some examples coming from \cite{LPVG1} and in the brief
  final section we draw some conclusions.
\newline
\newline
\noindent{\bf Acknowledgements}. The authors thank Antonio Moro for useful discussions. Paolo Lorenzoni, Sara Perletti  and Karoline van Gemst are supported by funds of  INFN   (Istituto Nazionale di Fisica Nucleare) by IS-CSN4 Mathematical Methods of Nonlinear Physics. The authors are thankful to GNFM (Gruppo Nazionale di Fisica Matematica) for supporting activities that contributed to the research reported in this paper.  

\section{$F$-manifolds and eventual identities}
\subsection{Regular $F$-manifolds}
$F$-manifolds were introduced by Hertling and Manin in \cite{HM} and are defined as follows.
\begin{definition}\label{defFmani}
An \emph{$F$-manifold} is a manifold $M$ equipped with
\begin{itemize}
\item[(i)] a commutative associative bilinear product  $\circ$  on the module of (local) vector fields, satisfying the \textit{Hertling--Manin condition}:
\begin{equation}\label{HM}
\mathcal{L}_{X \circ Y} \circ=X \circ (\mathcal{L}_Y \circ) +Y \circ (\mathcal{L}_X \circ ),
\end{equation}
for all local vector fields $X, Y$;
\newline
\item[(ii)] a distinguished vector field $e$ on $M$ such that 
\[e\circ X=X\] 
for all local vector fields $X$.
\end{itemize}
\end{definition}

$F$-manifolds are often equipped with additional structures. A particularly useful class of $F$-manifolds, for studying  systems of hydrodynamic type, is that of 
$F$-manifolds with compatible connection and flat unit.
\begin{definition}[\cite{LPR}]\label{Fmnfwcompatconn_flatunit}
An \emph{$F$-manifold with compatible connection and flat unit} is a manifold $M$ equipped with a commutative, associative, and unital product 
\[\circ : TM \times TM \rightarrow TM,\] 
with a torsionless connection $\nabla$ such that
\begin{itemize}
\item[(i)] For $R$ denoting the Riemann tensor of $\nabla$,
\begin{equation}\label{rc-intri}
Z\circ R(W,Y)(X)+W\circ R(Y,Z)(X)+ Y \circ R(Z,W)(X)=0,
\end{equation}
for all local vector fields $X$, $Y$, $Z$, $W$;
\item[(ii)] Let $c^k_{ij}$ denote the structure coefficients of the product $\circ$. Then
\begin{equation*}
     \nabla_j c^i_{hk} = \nabla_h c^{i}_{jk}.
\end{equation*}
\item[(iii)] The identity of the product $e$ is flat: $\nabla e=0$.
\end{itemize}
\end{definition}
$F$-manifolds with compatible connection naturally encode systems of hydrodynamic type:
\begin{equation*}
u^i_t=W^i_j(u)u^j_x,\qquad i\in \{1, \cdots, n\}.
\end{equation*}
In the diagonalisable case, any such system can be written in the form
\begin{equation}
u^i_t= (\mu \circ u_x)^i = c^i_{jk}\mu^j u^k_x,\qquad i\in \{1, \cdots, n\},
\label{eq:hydroprod}
\end{equation}
where $\mu$ is a local vector field  and $\circ$ is the product of a semisimple $F$-manifold. Thus, Riemann invariants coincide with Dubrovin's canonical coordinates in which the structure coefficients are given by
\[
c^i_{jk}=\delta^i_j\delta^i_k.
\]
Moreover, there is a unique compatible torsionless connection satisfying 
$d_{\nabla}(\mu \circ) = 0$ (and $\nabla e =0$), giving the structure of an $F$-manifold with compatible connection, where condition \eqref{rc-intri} encodes the integrability of the hydrodynamic-type system.

To pass from the semisimple to the regular non-semisimple setting, one needs the notion of an $F$-manifold with Euler vector field.
\begin{definition}\label{defFwithE}
An \emph{$F$-manifold with Euler vector field} is an $F$-manifold $M$ equipped with a vector field $E$ satisfying
\begin{equation}
	\mathcal{L}_E \circ=\circ.\label{Euler}
\end{equation}
\end{definition}

\begin{definition}[\cite{DH}]\label{DavidHertlingdef}
An $F$-manifold with Euler vector field $(M,\circ, e,E)$ is called \emph{regular} if for each $p\in M$ the matrix representing the endomorphism
$$L_p := E_p\circ : T_pM \to T_pM$$
has exactly one Jordan block for each distinct eigenvalue.
\end{definition}
If the Jordan normal form of $L=E\circ$ consists of $r$ blocks of sizes
$m_1,\dots,m_r$, we label coordinates by $j(\alpha)$, where
\[
\alpha\in\{1,\dots,r\}, \qquad j\in\{1,\dots,m_\alpha\}.
\]
Thus $u^{j(\alpha)}$ denotes the $j^{\text{th}}$ coordinate in the $\alpha^{\text{th}}$ Jordan block.
In other words, we may identify $u^{j(\alpha)}$ with a single index as
\[
j(\alpha)=m_1+\cdots+m_{\alpha-1}+j.
\]
The regular setting reduces to the semisimple one when $r=n$, in which case one retrieves systems of hydrodynamic type admitting a diagonal representation.

\begin{theorem}[\cite{DH}]\label{DavidHertlingth}
Let $(M, \circ, e, E)$ be a regular $F$-manifold of dimension $n \geq 2$ with an Euler vector field $E$. Furthermore, assume that locally around a point $p\in M$, the Jordan canonical form of the operator $L$ has $r$ Jordan blocks of sizes $m_1 , \dots, m_r$ with distinct eigenvalues. Then
there exists, locally around $p$, a distinguished system of coordinates $\{u^1, \dots, u^{m_1+\dots+m_r}\}$ such that
 \begin{align}
	e^{i(\alpha)}&=\delta^i_1,\qquad
	E^{i(\alpha)}=u^{i(\alpha)},\qquad
	c^{i(\alpha)}_{j(\beta)k(\gamma)}=\delta^\alpha_\beta\delta^\alpha_\gamma\delta^i_{j+k-1},\notag
\end{align}
for all $\alpha,\beta,\gamma\in\{1,\dots,r\}$ and  $i\in\{1,\dots,m_\alpha\}$, $j\in\{1,\dots,m_\beta\}$, ${k\in\{1,\dots,m_\gamma\}}$.
\end{theorem}

In the regular non-diagonalisable case, we consider systems of hydrodynamic type of the form \eqref{eq:hydroprod}
where $c^i_{jk}$ now denote the structure coefficients of a regular $F$-manifold in the David--Hertling coordinates of Theorem \ref{DavidHertlingth}. In David--Hertling coordinates the defining matrix $W=\mu \circ$ takes a block-diagonal form, $W=\text{diag}(W_{(1)} , \dots, W_{(r)})$, $r\leq n$, whose generic $\alpha^{\text{th}}$ block, of size $m_\alpha$, is of the lower-triangular Toeplitz form
\begin{equation}\label{toeplitz}
	W_{(\alpha)}=
	\begin{bmatrix}
		\mu^{1(\alpha)} & 0 & \dots & 0\cr
		\mu^{2(\alpha)} & \mu^{1(\alpha)} & \dots & 0\cr
		\vdots & \ddots & \ddots & \vdots\cr
		\mu^{m_\alpha(\alpha)} & \dots & \mu^{2(\alpha)} & \mu^{1(\alpha)}
	\end{bmatrix}.
\end{equation}

 In \cite{LPVG1} it was shown that, under the mild technical assumption that ${\mu^{1(\alpha)}\neq \mu^{1(\beta)}}$ for distinct blocks $\alpha, \beta$,  and that $\mu^{2(\alpha)}\neq0$  for every $\alpha$ such that  $m_\alpha\geq2$, this determines a unique torsionless connection $\nabla$ satisfying
\[
\nabla e = 0, \qquad d_\nabla (\mu \circ)=0,
\]
thus endowing the manifold with the structure of an $F$-manifold with compatible connection and flat unit. More precisely, we have the following theorem.

\begin{theorem}[\cite{LPVG1}]\label{LPVG_mainTh}
	Let $\{\mu_{(0)},\dots,\mu_{(n-1)}\}$ be a set of linearly independent local vector fields on an $n$-dimensional regular $F$-manifold $(M,\circ,e)$. Let us assume that the corresponding flows
	\begin{align}
		{\bf  u}_{t_i}=\mu_{(i)}\circ {\bf u}_x,\qquad i\in\{0,\dots,n-1\},
		\notag
	\end{align}
	pairwise commute, and that there exists a local vector field $\mu\in\{\mu_{(0)},\dots,\mu_{(n-1)}\}$ such that ${\mu^{1(\alpha)}\neq \mu^{1(\beta)}}$ for ${\alpha\neq\beta}$, $\alpha,\beta\in\{1,\dots,r\}$, and $\mu^{2(\alpha)}\neq0$  for each ${\alpha\in\{1,\dots,r\}}$ with $m_\alpha\geq2$. Then an $F$-manifold $(M,\circ,e,\nabla)$ with compatible connection and flat unit is defined by the unique torsionless connection $\nabla$ realising $\nabla e=0$ and $d_\nabla(\mu\circ)=0$.
\end{theorem}

In particular, fixing the connection by the condition $d_\nabla(\mu\circ)=0$ implies that $d_\nabla(\lambda \circ)=0$ for any other local vector field $\lambda$ whose flow commutes with that defined by $\mu$. This fact will be important in the subsequent sections. Furthermore, in \cite{LPVG3} it was shown that, in the setting of Theorem \ref{LPVG_mainTh},  the system for densities of conservation laws is compatible and $V$ is a density of conservation law if and only if $V$ satisfies
\begin{equation*}
    (\mu\circ)^s_i \nabla_s (d V)_j = (\mu \circ)^s_j \nabla_s(dV)_i.
\end{equation*}
\subsection{Eventual  identities}
The Euler vector field is a special case of a more general object --- an \emph{eventual identity}. 
\begin{definition}[\cite{M}]
    An eventual identity $\mathcal{E}$ for an $F$-manifold $(M, \circ, e)$ is an invertible vector field $\mathcal{E}$ such that the product
  \begin{equation}
        X \star Y  = \mathcal{E}^{-1} \circ X \circ Y,
        \label{eq:evId}
    \end{equation}
    defines an $F$-manifold $(M, \star, \mathcal{E})$. 
\end{definition}

\begin{remark}
\label{rmk:evID}
    It was proved in \cite{DS} that an invertible vector field $\mathcal{E}$ is an eventual identity if and only if 
\begin{equation}
     \mathcal{L}_{\mathcal{E}}(\circ)(X,Y) = [e, \mathcal{E}] \circ X \circ Y.
        \label{eq:evIdG}
\end{equation}
Note that \eqref{eq:evIdG} is precisely the coordinate-free formulation of the condition \eqref{EvId}, associated to constant densities $V$. 
\end{remark}
Eventual identities were studied in the setting of regular $F$-manifolds in \cite{PS}. 

In the regular case, the condition \eqref{eq:evIdG} can also be written in an equivalent  form in terms of the torsionless connection $\nabla$ defined by the conditions
\[
\nabla e = 0, \qquad d_\nabla (\mathcal{E} \circ)=0,
\]
\begin{proposition}\label{NablaEvId}
Let $(M, \circ,e,E)$ be a regular $F$-manifold.  A vector field $\mathcal{E}$ satisfying the conditions
\[\mathcal{E}^{1(\alpha)}\neq\mathcal{E}^{1(\beta)},\qquad\alpha\neq\beta, \alpha,\beta\in\{1,\dots,r\},\]
and 
\[\mathcal{E}^{2(\alpha)}\neq0,\qquad {\alpha\in\{1,\dots,r\}, \, m_{\alpha}\geq 2},\]
is an eventual identity if and only if
\begin{equation*}
\nabla_{\mathcal{E}}\circ=0.
\end{equation*}
\end{proposition}

\begin{proof}
This result follows from the identity
\begin{equation}\label{identity}
\nabla_{\mathcal{E}}c^i_{jh}=\mathcal{L}_{\mathcal{E}}c^i_{hj}-c^i_{js}c^s_{ht}[e,\mathcal{E}]^t.
\end{equation}  

Indeed, taking into account that, by definition of $\nabla$, we have
\begin{equation*}
\partial_j\mathcal{E}^i=-\Gamma^i_{js}\mathcal{E}^s+c^i_{js}\nabla_e\mathcal{E}^s=-\Gamma^i_{js}\mathcal{E}^s+c^i_{js}[e,\mathcal{E}]^s,
\end{equation*}
we obtain
\begin{eqnarray*}
\mathcal{L}_{\mathcal{E}}c^i_{hj}&=&\mathcal{E}^s\partial_sc^i_{hj}-c^s_{hj}\partial_s\mathcal{E}^i
+c^i_{sj}\partial_h\mathcal{E}^s+c^i_{hs}\partial_j\mathcal{E}^s\\
&=&\mathcal{E}^s\partial_sc^i_{hj}-c^s_{hj}(-\Gamma^i_{st}\mathcal{E}^t+c^i_{st}[e,\mathcal{E}]^t)
+c^i_{sj}(-\Gamma^s_{ht}\mathcal{E}^t+c^s_{ht}[e,\mathcal{E}]^t)+c^i_{hs}(-\Gamma^s_{jt}\mathcal{E}^t+c^s_{jt}[e,\mathcal{E}]^t)\\
&=&\nabla_{\mathcal{E}}c^i_{jh}+c^i_{js}c^s_{ht}[e,\mathcal{E}]^t.
\end{eqnarray*}
\end{proof}

The Nijenhuis torsion of a $(1,1)$ tensor field $L$ is given by
\begin{equation}
    N_L(X,Y) = [LX, LY] - L[LX, Y] - L[X, LY] + L^2[X, Y]
    \label{eq:Nijenhuis}
\end{equation}
where $X,Y$ are vector fields. 

\begin{theorem}
    Let $\mathcal{E}$ be a vector field on a regular $F$-manifold $(M, \circ, e, E)$ and let the Jordan canonical form of $\mathcal{E}\circ$ consists of $r$ Jordan blocks. Assume that  $\mathcal{E}^{1(\alpha)} \neq \mathcal{E}^{1(\beta)}$, for all $\alpha \neq \beta$, $\alpha, \beta \in \{1, \cdots, r\}$, and that $\mathcal{E}^{2(\alpha)} \neq 0$, for all $\alpha \in \{1, \cdots, r\}$ such that $m_{\alpha} \geq 2$.  Then, $\mathcal{E}$ is an eventual identity if $N_{\mathcal{E} \circ} = 0$. 
    \label{thm:NTevID}\end{theorem}

 By regularity, we can work in David--Hertling canonical coordinates, in which we know, from the form of the $\circ$ in this frame,  that $\mathcal{E} \circ$ is block-diagonal with each block having a lower-triangular Toeplitz form. We will be using the convention that out-of-range indices give zero contribution throughout.  To prove Theorem \ref{thm:NTevID}, we rely on the following lemmas. 

\begin{lemma}
   \label{lemma:evidsplit}
   Let $\circ$ be the product associated to a regular $F$-manifold, and let $\mathcal{E}$ be a vector field for which the Jordan canonical form of $\mathcal{E} \circ $ consists of $r$ blocks. Let $\mathcal{E}^{1(\alpha)} \neq \mathcal{E}^{1(\beta)}$ for any distinct pair $\alpha, \beta \in \{1, \cdots, r\}$.   If $N_{\mathcal{E} \circ} = 0$, then in David--Hertling canonical coordinates, $\mathcal{E}$ has the form
   \begin{equation*}
       \mathcal{E} = \sum_{\alpha =1}^{r}\mathcal{E}_{(\alpha)}, \quad \text{where } \, \mathcal{E}_{(\alpha)} = \sum_{i=1}^{m_{\alpha}}\mathcal{E}^{i(\alpha)}(u^{1(\alpha)}, \cdots, u^{m_{\alpha}(\alpha)})\partial_{i(\alpha)}
   \end{equation*}
   where $m_{\alpha}$ is the size of block $\alpha$. 
\end{lemma}

\begin{proof}
    Consider \eqref{eq:Nijenhuis}, with $X = \partial_{i(\alpha)}$, $Y=\partial_{j(\beta)}$, and $L=\mathcal{E}\circ$, where $\mathcal{E} = \sum_{\alpha=1}^{r}\sum_{i=1}^{m_{\alpha}}\mathcal{E}^{i(\alpha)}\partial_{i(\alpha)}$,
    for $\alpha, \beta \in \{1, \cdots, r\}$ with $\alpha \neq \beta$. Note that, by the form of the structure coefficients in Theorem \ref{DavidHertlingth},
    \begin{equation*}
        LX = \mathcal{E} \circ X = \sum_{\gamma=1}^{r}\sum_{k=1}^{m_{\gamma}}\mathcal{E}^{k(\gamma)}\partial_{k(\gamma)} \circ \partial_{i(\alpha)}  = \sum_{k=1}^{m_{\alpha}-i+1}\mathcal{E}^{k(\alpha)}\partial_{(k+i-1)(\alpha)}. 
    \end{equation*}
    Thus, 
    \begin{align*}
        N_L(X,Y) = \, & \, \left[\sum_k \mathcal{E}^{k(\alpha)}\partial_{(k+i-1)(\alpha)}, \sum_{l}\mathcal{E}^{l(\beta)}\partial_{(l+j-1)(\beta)}\right] - L\left[\sum_{k}\mathcal{E}^{k(\alpha)}\partial_{(k+i-1)(\alpha)}, \partial_{j(\beta)}\right] - \\  \, & \, - L\left[\partial_{i(\alpha)}, \sum_{l}\mathcal{E}^{l(\beta)}\partial_{(l+j-1)(\beta)}\right]  = T_{\alpha}+T_{\beta}, 
    \end{align*}
    where
    \begin{equation*}
        T_{\alpha} = -\sum_{k,l}\mathcal{E}^{l(\beta)}\partial_{(l+j-1)(\beta)}(\mathcal{E}^{k(\alpha)})\partial_{(k+i-1)(\alpha)} + \sum_{k,m} \mathcal{E}^{m(\alpha)}\partial_{j(\beta)}(\mathcal{E}^{k(\alpha)})\partial_{(m+k+i-2)(\alpha)},
    \end{equation*}
    and $T_{\beta} = -T_{\alpha}$ after letting $\alpha \leftrightarrow \beta$, $i \leftrightarrow j$. 
    Note that $N_L(X,Y) = 0$ if and only if $T_{\alpha}$ and $T_{\beta}$  vanish independently.  Thus, it suffices to consider $T_{\alpha}$. The coefficient of $\partial_{s(\alpha)}$ in $T_{\alpha}$ is given by
    \begin{equation*}
        [\partial_{s(\alpha)}]T_{\alpha} = -\sum_{l=1}^{m_{\beta}-j+1}\mathcal{E}^{l(\beta)}\partial_{(l+j-1)(\beta)}(\mathcal{E}^{(s-i+1)(\alpha)})+\sum_{m=1}^{s-i+1}\mathcal{E}^{m(\alpha)}\partial_{j(\beta)}(\mathcal{E}^{(s-i+2-m)(\alpha)}),
    \end{equation*}
    and letting $p\coloneqq s-i+1$ gives
    \begin{equation}
        T_{\alpha} = 0 \implies (\mathcal{E}^{1(\alpha)}-\mathcal{E}^{1(\beta)})\partial_{j(\beta)}\mathcal{E}^{p(\alpha)} = \sum_{l=2}^{m_{\beta}-j+1}\mathcal{E}^{l(\beta)}\partial_{(l+j-1)(\beta)}\mathcal{E}^{p(\alpha)} - \sum_{m=2}^{p}\mathcal{E}^{m(\alpha)}\partial_{j(\beta)}\mathcal{E}^{(p+1-m)(\alpha)}. 
        \label{eq:splittemp1}
    \end{equation}
    We will now perform double induction over $p$ (ascending) and $j$ (descending). 
    \begin{itemize}
        \item Let $p=1$. Then \eqref{eq:splittemp1} becomes
        \begin{equation}
            (\mathcal{E}^{1(\alpha)}-\mathcal{E}^{1(\beta)})\partial_{j(\beta)} \mathcal{E}^{1(\alpha)}= \sum_{l=2}^{m_{\beta}-j+1}\mathcal{E}^{l(\beta)} \partial_{(l+j-1)(\beta)}\mathcal{E}^{1(\alpha)}.
            \label{eq:splittemp2}
        \end{equation}
        Let $j=m_{\beta}$.  Then \eqref{eq:splittemp2} gives
        \begin{equation*}
             (\mathcal{E}^{1(\alpha)}-\mathcal{E}^{1(\beta)})\partial_{m_{\beta}(\beta)} \mathcal{E}^{1(\alpha)}=0,
        \end{equation*}
        since $(l+j-1)(\beta) = (l-1+m_{\beta})(\beta) \geq (m_{\beta}+1)(\beta)$ (as $l \geq 2$), which implies
        \begin{equation*}
            \partial_{m_{\beta}(\beta)}\mathcal{E}^{1(\alpha)} = 0, 
        \end{equation*}
        by the non-coincidence assumption on the first components.
      
        Now let $\partial_{q(\beta)}\mathcal{E}^{1(\alpha)} = 0$, $\forall \, q>j$. Then since $l+j-1>j$, $\partial_{(l+j-1)(\beta)}\mathcal{E}^{1(\alpha)} = 0$, $\forall \, l \geq 2$ and 
        \begin{equation*}
            (\mathcal{E}^{1}(\alpha)-\mathcal{E}^{1(\beta)})\partial_{j(\beta)}\mathcal{E}^{1(\alpha)} = 0 \implies \partial_{j(\beta)}\mathcal{E}^{1(\alpha)} = 0, \quad \forall \, j, \, \alpha \neq \beta. 
        \end{equation*}
        \item Let us now assume that for $q<p$,
        \begin{equation*}
            \partial_{j(\beta)}\mathcal{E}^{q(\alpha)} = 0, \quad \forall \, j \in \{1, \cdots, m_{\beta}\}, \, \alpha \neq \beta.
        \end{equation*}
        Then \eqref{eq:splittemp1} gives
        \begin{equation}
            (\mathcal{E}^{1}(\alpha)-\mathcal{E}^{1(\beta)})\partial_{j(\beta)}\mathcal{E}^{p(\alpha)} = \sum_{l =2}^{m_{\beta}-j+1}\mathcal{E}^{l(\beta)}\partial_{(l+j-1)(\beta)}\mathcal{E}^{p(\alpha)}
            \label{eq:splittemp3}
        \end{equation}
       by the induction hypothesis as $p+1-m \leq p-1$. Let us again perform induction over (descending) $j$. 
      
       Let $j=m_{\beta}$. Then the right-hand-side of \eqref{eq:splittemp3} vanishes, since $l+j-1 = l+m_{\beta}-1 > m_{\beta}$, giving the result by the non-coincidence assumption on the first components. 
       
       Let us now assume that the statement holds for all $q>j$. Again, the right-hand-side of \eqref{eq:splittemp3} vanishes, since $l+j-1>j$ since $l \geq 2$, giving
       \begin{equation*}
           (\mathcal{E}^{1(\alpha)}-\mathcal{E}^{1(\beta)})\partial_{j(\beta)}\mathcal{E}^{p(\alpha)} =0 \implies \partial_{j(\beta)}\mathcal{E}^{p(\alpha)} = 0, 
       \end{equation*}
       as required. 
    \end{itemize}
    Therefore, each coefficient $\mathcal{E}^{p(\alpha)}$ depends only on the coordinates of the block $\alpha$, $\{u^{1(\alpha)}, \cdots, u^{m_{\alpha}(\alpha)}\}$, concluding the proof of the lemma. 
    
\end{proof}

\begin{lemma}
    \label{lemma:splittorsion}
For $\mathcal{E}$ as in Lemma \ref{lemma:evidsplit}, 
\begin{equation*}
    N_{\mathcal{E} \circ} = 0 \Leftrightarrow N_{\mathcal{E}_{(\alpha)}\circ} = 0, \quad \forall \, \alpha \in \{1, \cdots, r\}. 
\end{equation*}
\end{lemma}

\begin{proof}
    Suppose 
    \begin{equation*}
        X = \sum_{i =1}^{m_{\alpha}}f^i(u^{1(\alpha)}, \cdots, u^{m_{\alpha}(\alpha)})\partial_{i(\alpha)}, \quad  Y = \sum_{j =1}^{m_{\beta}}g^j(u^{1(\beta)}, \cdots, u^{m_{\beta}(\beta)})\partial_{j(\beta)},
    \end{equation*}
    then 
    \begin{equation*}
        [X,Y] = \sum_{i,j}^{}\left(f^i \partial_{i(\alpha)}(g^j)\partial_{j(\beta)} - g^j \partial_{j(\beta)}(f^i)\partial_{i(\alpha)}\right) = 
           \sum_{i,j}^{}\left(f^i\partial_{i(\alpha)}(g^j)\partial_{j(\alpha)}-g^j\partial_{j(\alpha)}(f^i)\partial_{i(\alpha)}\right),
    \end{equation*}
    if $\alpha = \beta$, and zero otherwise. 
    Moreover, 
\begin{equation*}
    L(\partial_{i(\alpha)}) = \mathcal{E} \circ \partial_{i(\alpha)} = \sum_{\beta = 1}^{r}\mathcal{E}_{(\beta)}  \circ \partial_{i(\alpha)} = 
        \sum_{k=1}^{m_{\alpha}}\mathcal{E}^{k(\alpha)}\partial_{(k+i-1)(\alpha)}.
\end{equation*} 
Hence, if $\alpha \neq \beta$, all four terms in $N_L(X,Y)$ vanish, and if $\alpha = \beta$, then $N_L(X,Y) = N_{\mathcal{E}_{(\alpha)}\circ}(X,Y)$. By bilinearity, the result follows. 
\end{proof}

In the following lemma we will be dealing with a single fixed block, $\alpha \in \{1, \cdots, r\}$. For brevity we will write $i$ for $ i(\alpha)$, let $n \coloneqq m_{\alpha}$, and $\mathcal{E} = \sum_{i=1}^{n}\mathcal{E}^i \partial_i$. 
\begin{lemma}
\label{lemma:basis}
Let $\mathcal{E}^{\circ i}$ denote the $i$-fold product of $\mathcal{E}$ with itself under $\circ$.  If $\mathcal{E}^{2} \neq 0$, then
    \begin{equation*}
        \{\mathcal{E}^{\circ i}\}_{i=0, \cdots n-1}
    \end{equation*}
    spans the associated block as a vector space. 
\end{lemma}
\begin{proof}
    The lemma will follow from showing that for $F \coloneqq \mathcal{E}-\mathcal{E}^1 e = \sum_{a=2}^{n}\mathcal{E}^a \partial_a$
    \begin{equation*}
        \{F^{\circ k}\}_{k=0, \cdots n-1}, 
    \end{equation*}
    form a basis of the block as a vector space, together with the fact that each $F^{\circ k}$ is a polynomial in $\mathcal{E}$. 
Let us first prove the intermediate result. For each $k \in \{0, \cdots, n-1\}$, there exists smooth functions $c_{k,s}$ such that
\begin{equation*}
    F^{\circ k} = (\mathcal{E}^2)^k \partial_{k+1}+ \sum_{s=k+2}^{n}c_{k,s}\partial_s. 
\end{equation*}
    This follows by induction on $k$. The statement is immediate for $k=0,1$ since $e=\partial_1$ and by definition of $F$. Thus,  assume that the statement holds for some $k \in \{1, \cdots, n-2\}$. Multiplying both sides by $F$ gives
\begin{equation*}
    F^{\circ(k+1)} = \left((\mathcal{E}^2)^k \partial_{k+1}+ \sum_{s=k+2}^{n}c_{k,s}\partial_s\right)\circ \left(\sum_{a=2}^{n}\mathcal{E}^a \partial_a\right). 
\end{equation*}
Now, since
\begin{equation*}
    \partial_i \circ \partial_j = \partial_{i+j-1}
\end{equation*}
we have that 
\begin{equation*}
    (\mathcal{E}^2)^k \partial_{k+1} \circ \mathcal{E}^2 \partial_2 = (\mathcal{E}^2)^{k+1}\partial_{k+2}, 
\end{equation*}
and that every other term lies in the span of $\{\partial_{k+3}, \cdots \partial_n\}$, giving
\begin{equation*}
    F^{\circ(k+1)} = (\mathcal{E}^2)^{k+1} \partial_{k+2} + \sum_{s=k+3}^{n}c_{k+1, s}\partial_s, 
\end{equation*}
as required.  This shows that for each $k$, 
\begin{equation*}
    F^{\circ k} \in \operatorname{Span}\{\partial_{k+1}, \dots, \partial_n\}
\end{equation*}
with leading term $(\mathcal{E}^2)^k \partial_{k+1}$. It follows that the expansion of $\{F^{\circ i}\}_{i=0, \cdots, n-1}$ in the basis $\{\partial_i\}$ is upper triangular, with diagonal entries $1, \mathcal{E}^2, \cdots, (\mathcal{E}^2)^{n-1}$. Thus, under the assumption $\mathcal{E}^2 \neq 0$, all diagonal entries are non-zero, and the change-of-basis matrix is invertible. Finally, since $F=\mathcal{E}-\mathcal{E}^1 e$ and $e$ is the unit, each $F^{\circ k}$ is a polynomial in $\mathcal{E}$  which implies
\begin{equation*}
    \operatorname{Span}\{F^{\circ i}\}_{i=0, \cdots, n-1} \subseteq \operatorname{Span}\{\mathcal{E}^{\circ i}\}_{i=0, \cdots, n-1}. 
\end{equation*}
Since $\{F^{\circ i}\}_{i=0}^{n-1}$ forms a basis of the block,  $\operatorname{Span}\{\mathcal{E}^{\circ i}\}_{i=0}^{n-1}$ contains the block, and therefore spans the block.
\end{proof}

We now proceed to prove Theorem \ref{thm:NTevID}. 

\begin{proof}
    By Lemma \ref{lemma:splittorsion}, it is sufficient to consider a single Jordan block $\alpha \in \{1, \cdots, r\}$. As above, we will in this proof let $i \equiv i(\alpha)$, $n \coloneqq m_{\alpha}$, and $\mathcal{E} = \sum_{i=1}^{n}\mathcal{E}^i \partial_i$. 

Define
\begin{equation*}
    K(X,Y) \coloneqq \mathcal{L}_{\mathcal{E}}(\circ)(X,Y) - [e,\mathcal{E}] \circ X \circ Y. 
\end{equation*}
Then, by Remark \ref{rmk:evID}, $\mathcal{E}$ is an eventual identity if and only if
\begin{equation*}
    K(X,Y) = 0, \quad \forall X, Y. 
\end{equation*}
    Note that, 
    \begin{align*}
        \mathcal{L}_{\mathcal{E}}(\circ)(\partial_i, \partial_j) \, & \, = [\mathcal{E}, \partial_i \circ \partial_j] - [\mathcal{E}, \partial_i]\circ \partial_j - \partial_i \circ [\mathcal{E}, \partial_j]\\
        \, & \, = -\sum_{t}^{}\partial_{i+j-1}(\mathcal{E}^t)\partial_t + \sum_{a}^{}\partial_i(\mathcal{E}^a)\partial_{a+j-1}+\sum_{a}^{}\partial_j(\mathcal{E}^a)\partial_{a+i-1}\\
        \, & \, = \sum_{t=1}^{n} \left(-\partial_{i+j-1}\mathcal{E}^t + \partial_i \mathcal{E}^{t-j+1}+\partial_j \mathcal{E}^{t-i+1} \right) \partial_t. 
    \end{align*}
    Moreover, 
    \begin{equation*}
        [e, \mathcal{E}] \circ \partial_i \circ \partial_j = \sum_{a=1}^n\partial_1(\mathcal{E}^a) \partial_{a+i+j-2} = \sum_{t=1}^{n}\partial_1(\mathcal{E}^{t-i-j+2})\partial_t,
    \end{equation*}
    where the second equality follows after re-indexing $t=a+i+j-2$. 
    Thus,
    \begin{equation}
        K(\partial_i, \partial_j) = \sum_{t=1}^{n}\left(-\partial_{i+j-1}\mathcal{E}^{t} + \partial_i \mathcal{E}^{t-j+1} + \partial_j \mathcal{E}^{t-i+1} - \partial_1 \mathcal{E}^{t-i-j+2}\right)\partial_t.
        \label{eq:Nevidtemp1}
    \end{equation}
Let us show that 
\begin{equation*}
    N_L(X,Y) = K(LX, Y) - K(X, LY). 
\end{equation*}
Firstly, 
\begin{equation*}
    K(L \partial_i, \partial_j) - K(\partial_i, L \partial_j) = \sum_{a=1}^{n}\mathcal{E}^a(K(\partial_{a+i-1}, \partial_j) - K(\partial_i, \partial_{a+j-1})),
    \end{equation*}
    which means that, using \eqref{eq:Nevidtemp1}, $ [\partial_t](K(L\partial_i, \partial_j) - K(\partial_i, L \partial_j))$ is given by
\begin{equation}
 \sum_{a=1}^{n}\mathcal{E}^a\left(\partial_{a+i-1} \mathcal{E}^{t-j+1} - \partial_{a+j-1}\mathcal{E}^{t-i+1} + \partial_j \mathcal{E}^{t-a-i+2} - \partial_i \mathcal{E}^{t-a-j+2}\right).
    \label{eq:Nevidtemp2}
\end{equation}
On the other hand, since $[\partial_i, \partial_j] = 0$, 
\begin{equation*}
    N_L(\partial_i, \partial_j) = [L\partial_i, L\partial_j] - L[L\partial_i, \partial_j] - L[\partial_i, L \partial_j], 
\end{equation*}
where 
$[L\partial_i, L\partial_j]$ reads
    \begin{equation*}
        \sum_{a,b}^{}\left(\mathcal{E}^a \partial_{a+i-1}(\mathcal{E}^b)\partial_{b+j-1} - \mathcal{E}^b\partial_{b+j-1}(\mathcal{E}^a)\partial_{a+i-1}\right),
    \end{equation*}
    which gives (by letting $b+j-1=t$, and $a+i-1=t$ in the first and second sums respectively, and collecting)
    \begin{equation*}
        [\partial_t][L\partial_i, L\partial_j] = \sum_{s=1}^{n}\mathcal{E}
^s\left(\partial_{s+i-1}\mathcal{E}^{t-j+1}-\partial_{s+j-1}\mathcal{E}^{t-i+1}\right).    \end{equation*}
    Similarly,
    \begin{equation*}
        [\partial_t](-L[L\partial_{i}, \partial_j]) = \sum_{s=1}^{n}\mathcal{E}^s \partial_j\mathcal{E}^{t-s-i+2}, \quad [\partial_t](-L[\partial_{i}, L\partial_j]) = -\sum_{s=1}^{n}\mathcal{E}^s \partial_i\mathcal{E}^{t-s-j+2}.
    \end{equation*}
Combining gives precisely \eqref{eq:Nevidtemp2} (after renaming $s \mapsto a$). Hence, 
    \begin{equation*}
    N_L(X,Y) = K(LX, Y) - K(X, LY),
\end{equation*}
    and, if $N_L(X,Y) = 0$ we have
    \begin{equation}
        K(LX, Y) = K(X, LY),  \quad  \forall \, X, Y.
        \label{eq:Nevidtemp3}
    \end{equation}

We will now use this fact together with Lemma \ref{lemma:basis} to complete the proof of the theorem.

Note that $L^p e = \mathcal{E}^{\circ p}$. This follows with a brief induction over $p$. For $p=0$ we have $L^0 e = e$ trivially. Assuming the identity holds for $p$, then
\begin{equation*}
    L^{p+1}e  = L(L^p e) = \mathcal{E} \circ \mathcal{E}^{\circ p} = \mathcal{E}^{\circ (p+1)}. 
\end{equation*}
Now, let $p,q \geq 0$. By repeatedly applying \eqref{eq:Nevidtemp3} we obtain a chain of equalities
\begin{equation*}
    K(L^pe, L^qe) = K(L^{p-1}e, L^{q+1}e) = \cdots = K(e, L^{p+q}e) = 0.
\end{equation*}
The final equality follows from the fact that 
 \begin{align*}
       \mathcal{L}_{\mathcal{E}}(\circ)(e,X)  \, & \, = [\mathcal{E}, e \circ X] - [\mathcal{E}, e]\circ X - e \circ [\mathcal{E},X]  = [\mathcal{E},X] + [e, \mathcal{E}] \circ X - [\mathcal{E},X] = [e, \mathcal{E}] \circ X \\ \, & \, = [e, \mathcal{E}] \circ e \circ X \implies K(e,X) = 0.
    \end{align*}
Therefore, 
\begin{equation*}
    K(L^pe, L^q e) = 0 \quad \forall \, p,q \geq 0 \implies K(\mathcal{E}^{\circ p}, \mathcal{E}^{\circ q}) = 0, \quad \forall \, p,q \geq 0. 
\end{equation*}
Since $\{\mathcal{E}^{\circ i}\}_{i=0, \cdots, n-1}$ spans the block by Lemma \ref{lemma:basis}, the theorem holds by the bilinearity of $K$ and its vanishing on the spanning set in both arguments. 
\end{proof}

From the results of \cite{ALimrn}, we have the following proposition.
\begin{proposition}\label{NiEvId}
    Let $L \equiv \mathcal{E} \circ$, where $\circ$ is the product associated to a regular $F$-manifold $(M, \circ, e, E)$,  $\mathcal{E}$ is a vector field, and  $\mathcal{E}\circ$ has Jordan canonical form consisting of $r$ blocks in David--Hertling coordinates. Assume that 
    \begin{equation*}
        \mathcal{E}^{1(\alpha)} \neq \mathcal{E}^{1(\beta)}, \quad \text{and} \quad \mathcal{E}^{2(\alpha)} \neq 0, 
    \end{equation*}
    for $\alpha, \beta \in \{1, \cdots, r\}$ such that  $\alpha \neq \beta$, and for the latter condition $m_{\alpha} \geq 2$. Then $N_L = 0$ if and only if $\mathcal{E}$ is an eventual identity. 
\end{proposition}

\subsection{A class of $F$-manifolds related to the Fr\"{o}licher--Nijenhuis  bicomplex}
Let $L$ be a  tensor field of type $(1,1)$ with vanishing Nijenhuis torsion. This means that for any pair of vector fields $X$ and $Y$, we have
\begin{equation}\label{nijenhis}
[LX,LY]-L\,[X,LY]-L\,[LX,Y]+L^2\,[X,Y]=0.
\end{equation}
Following \cite{LM}, we now recall a construction of integrable hierarchies starting from the  Fr\"{o}licher--Nijenhuis  bicomplex $(d,d_L,\Omega(M))$ (see \cite{FN}). 
  The differential $d$ is the usual de Rham differential, while the differential $d_L$ is defined as
  \begin{align}\label{eq:dLoriginal}
 (d_L \omega)(X_0, \dots, X_k):=&\sum_{i=0}^k (-1)^i (LX_i)(\omega(X_0, \dots, \hat{X}_i, \dots, X_k))\\
\nonumber &+\sum_{0\leq i<j\leq k}(-1)^{i+j}\,\omega([X_i, X_j]_L, X_0, \dots, \hat{X}_i, \dots, \hat{X}_j, \dots X_k),\end{align}
where  $\omega\in\Omega^k(M)$ and  
$$[X_i,X_j]_L:=[LX_i, X_j]+[X_i, LX_j]-L[X_i,X_j].$$

In \cite{LM}, it was proved that the sequence of functions $a_1,a_2,a_3,...$  defined recursively by 
\[da_{k+1}=d_L a_k -a_k da_0,\qquad k\geq0,\]
starting from a solution of the equation
\begin{equation}
d\cdot d_La_0=0,
\label{eq:main}
\end{equation}
and the corresponding sequence of the tensor fields of type $(1,1)$  defined  by 
\begin{eqnarray*}
V_{k+1} \coloneqq V_k L-a_k I,\qquad k\geq0,
\end{eqnarray*}
starting from $V_0=I$ define an integrable hierarchy of hydrodynamic type. This construction can be applied in particular in the case where $L=\mu\circ$ is the operator  of  multiplication by an  eventual identity. Applying the results of the previous section one can define  an $F$-manifold with compatible  connection. This was done in \cite{LP23} in a special case leading to the definition of  Lauricella bi-flat $F$-manifolds  and  in the general case (choosing as eventual identity the Euler vector field) in \cite{LPVG1}.

\section{Generalised Gibbons--Tsarev system for regular non-diagonalisable reductions of dKP}
We look for  reductions of the dKP equation in potential form
\begin{equation}\label{pdKP}
U_{yy}-U_{xt}+U_xU_{xx}=0
\end{equation}
which can equivalently be written as (the first equation locally gives $V=U_x,W=U_y$):
\begin{eqnarray}
\label{dkp1}
V_y&=&W_x,\\
\label{dkp2}
W_y&=&V_t-VV_x.
\end{eqnarray}
We look for $n$-component reductions $V(r^1,...,r^n),W(r^1,...,r^n)$ governed by $n$-component compatible systems of  hydrodynamic type
\begin{eqnarray}
\label{y-flow}
r^i_y&=&c^i_{jk}(r)\mu^j(r)r^k_x,\\
\label{t-flow}
r^i_t&=&c^i_{jk}(r)\lambda^j(r)r^k_x,
\end{eqnarray}
where $c^i_{jk}(r)$ are the structure functions of a regular $F$-manifold structure with product $\circ$ and unit $e$. In the semisimple case these are the diagonal reductions of dKP
 and the present construction drastically simplifies and reduces to the computations given in \cite{FK}. 

From equation \eqref{dkp1} it follows that
\begin{equation}\label{red-dkp1}
\partial_iW=\partial_kV\,c^k_{ij}\mu^j
\end{equation}
and from equation \eqref{dkp2} it follows that
\begin{equation}\label{red-dkp2}
\partial_k W\,c^k_{ij}\mu^j=\partial_kV\,c^k_{ij}\lambda^j-V\partial_iV.
\end{equation}
Using the equation \eqref{red-dkp1} the previous equation reduces to
\begin{equation}\label{red-dkp2-bis}
c^s_{ij}\partial_sV(\mu\circ\mu-\lambda+Ve)^j=0.
\end{equation}
Thanks to regularity, this implies
\begin{equation}\label{lambda-mu}
\lambda=\mu\circ\mu+Ve.
\end{equation}
Equation \eqref{dkp1} means that $V$ is a density of conservation law for the system \eqref{y-flow} and a system \eqref{t-flow} compatible with \eqref{y-flow} is usually called a symmetry. According to the general theory developed in \cite{LPVG1,LPVG2,LPVG3} symmetries and densities of conservation laws can be obtained starting from a torsionless linear connection $\nabla$ uniquely determined by the condition
\begin{equation}\label{nabla}
d_{\nabla}(\mu\circ)=0,\qquad \nabla e=0.
\end{equation}
It turns out that the connection $\nabla$ satisfies the condition
\begin{equation}
\label{symnablaprod}
\nabla_jc^i_{hk}=\nabla_hc^i_{jk}.
\end{equation}
The fact that $V$ is a density of conservation law is equivalent to the system
\begin{equation}\label{cl}
(\mu\circ)^s_i\nabla_s(dV)_j=(\mu\circ)^s_j\nabla_s(dV)_i,
\end{equation}
while the fact that the flow \eqref{t-flow} is a symmetry of the flow \eqref{y-flow} is equivalent to the condition 
\begin{equation}\label{sym}
d_{\nabla}(\lambda\circ)=0.
\end{equation}
The condition \eqref{sym} in local coordinates reads
\begin{equation}\label{sym2}
\partial_j\lambda^i=-\Gamma^i_{js}\lambda^s+c^i_{js}\nabla_e\lambda^s
\end{equation}
or
\begin{equation}\label{sym3}
\partial_j\lambda^i=-\Gamma^i_{js}\lambda^s+c^i_{js}[e,\lambda]^s.
\end{equation}
From \eqref{lambda-mu} it follows that
\begin{equation}\label{lambda-mu-diff}
\partial_j\lambda^i=(\partial_jc^i_{hk})\mu^h\mu^k+2c^i_{hk}(\partial_j\mu^h)\mu^k+(\partial_je^i)V+e^i\partial_jV.
\end{equation}
Combining \eqref{sym3} and \eqref{lambda-mu-diff}, and taking into account \eqref{lambda-mu} we obtain
\begin{equation*}
(\partial_jc^i_{hk})\mu^h\mu^k+2c^i_{hk}(\partial_j\mu^h)\mu^k+(\partial_je^i)V+e^i\partial_jV
=-\Gamma^i_{js}(\mu\circ\mu+Ve)^s+c^i_{js}[e,\mu\circ\mu+Ve]^s.
\end{equation*}
or, taking into account $\nabla e=0$,
\begin{equation}\label{f1}
(\partial_jc^i_{hk}+\Gamma^i_{js}c^s_{hk})\mu^h\mu^k+2c^i_{hk}(\partial_j\mu^h)\mu^k+e^i\partial_jV
=c^i_{js}[e,\mu\circ\mu]^s+e(V)\delta^i_j.
\end{equation}
Thanks to the  hypothesis of regularity, we can work  in coordinates where the product and the unit are simultaneously constant.  In such coordinates, taking into account that
\begin{equation}\label{nablamu}
\partial_j\mu^h=-\Gamma^h_{js}\mu^s+c^h_{js}[e,\mu]^s,
\end{equation}
and that
\[c^i_{js}[e,\mu\circ\mu]^s=2c^i_{js}c^s_{lk}(\partial_e\mu^l)\mu^k,\]
we can write \eqref{f1} as
\begin{equation}\label{f1bis}
\Gamma^i_{js}c^s_{hk}\mu^h\mu^k-2c^i_{hk}\Gamma^h_{js}\mu^s\mu^k+e^i\partial_jV
=e(V)\delta^i_j.
\end{equation}
Using \eqref{symnablaprod}, this can also be written as
\begin{equation}\label{f2}
(\nabla_{\mu}c^i_{jh})\mu^h=e(V)\delta^i_j-e^i\partial_jV,
\end{equation}
which is valid in any system of coordinates. Finally, expanding  \eqref{cl}, one easily obtains
\begin{equation}\label{GT1}
c^s_{it}\mu^t\partial_s\partial_j V-c^s_{jt}\mu^t\partial_s\partial_i V=
(c^t_{sj}\partial_i\mu^s-c^t_{si}\partial_j\mu^s)\partial_tV.
\end{equation}
Taking into account the identity \eqref{identity}, we can summarise this result in the following
\begin{proposition}\label{Prop_dKP_red}
The functions $V(r^1,...,r^n),\mu^1(r^1,...,r^n),\dots,\mu^n(r^1,...,r^n)$ with
\begin{eqnarray}
\label{y-flow_bis}
r^i_y&=&c^i_{jk}(r)\mu^j(r)r^k_x,\\
\label{t-flow_bis}
r^i_t&=&c^i_{jk}(r)\lambda^j(r)r^k_x.
\end{eqnarray}
define an $n$-component reduction of the dKP equation in potential form \eqref{pdKP} if and only if $V$  has a non-trivial dependence on each variable, $\lambda=\mu\circ\mu+Ve$ and the system \eqref{GTgen} is satisfied.
\end{proposition}

\section{Non-existence of non-diagonalisable regular reductions of dKP}
We consider the generalised Gibbons--Tsarev system in David--Hertling canonical coordinates, where both the product and the unit vector field are constant:
\begin{align}
	\mn^{t}\bigl(c^{s}_{ti}\,\pa_{s}\pa_{j}V-c^{s}_{tj}\,\pa_{s}\pa_{i}V\bigr)
	-\bigl(c^{m}_{js}\,\pa_{i}\mn^{s}-c^{m}_{is}\,\pa_{j}\mn^{s}\bigr)\pa_{m}V&=0,
	\label{eq:EQ1gen}
	\\
	\mn^{h}\bigl(c^{i}_{hk}\,\pa_{j}\mn^{k}+c^{i}_{jk}\,\pa_{h}\mn^{k}-c^{k}_{jh}\,\pa_{k}\mn^{i}\bigr)
	-\mn^{h}c^{i}_{js}c^{s}_{hk}e^{m}\pa_{m}\mn^{k}
	+e^{i}\pa_{j}V-\del^{i}_{j}e^{m}\pa_{m}V&=0.
	\label{eq:EQ2gen}
\end{align}
We recall that we are working in the regular setting, where $\mu$ is assumed to be such that $\mu^{1(\alpha)}\neq\mu^{1(\beta)}$ for $\beta\neq\alpha$ and $\mu^{2(\alpha)}\neq0$ for $m_\alpha\geq2$.

\subsection{The one-block case}
In the extreme case of a single Jordan block, of size $n\geq2$, we have
\begin{equation}
	c^{i}_{jk}=\del^{i+1}_{j+k},
	\qquad
	e^{i}=\del^{i}_{1}.
	\label{eq:oneblock-constants}
\end{equation}
The following proposition shows that a solution $V$ of \eqref{eq:EQ1gen}--\eqref{eq:EQ2gen} must be constant.

\begin{proposition}\label{non_ex_one_block}
	Let $V$ and $\{\mu^i\}_{i\in\{1,\dots,n\}}$ solve \eqref{eq:EQ1gen}--\eqref{eq:EQ2gen}. Then
	\begin{itemize}
		\item[(i)] $\pa_jV=0,\quad j\in\{1,\dots,n\},$
		\item[(ii)] $\pa_j\mu^1=0,\quad j\in\{2,\dots,n\}.$
	\end{itemize}
\end{proposition}

\begin{proof}
	The proof naturally splits into multiple parts.
	
	\textsc{Step 1: $\pa_jV=0$, $j\ge2$.} Let us fix an index $j\in\{2,\dots,n\}$. By considering the sum over $s\in\{1,\dots,n-j+1\}$ of the equations \eqref{eq:EQ2gen} for $i=s$, replacing $j$ by $j+s-1$, we have
	\begin{equation}
		0=\pa_jV+\Sigma_1+\Sigma_2,
		\label{eq:step1-pre-sums}
	\end{equation}
	where $\pa_jV$ is the single term coming from the only contribution, for $s=1$, from the last two terms of \eqref{eq:EQ2gen}, and where $\Sigma_1$ and $\Sigma_2$ collect the terms containing derivatives of the functions $\mu^k$. Explicitly,
	\begin{align*}
		\Sigma_1&=\overset{n-j+1}{\underset{s=1}{\sum}}\mu^h\big(c^s_{hk}\partial_{j+s-1}\mu^k\big), \qquad \Sigma_2=-\overset{n-j+1}{\underset{s=1}{\sum}}\mu^h\big(c^k_{j+s-1,h}\partial_{k}\mu^s\big).
	\end{align*}
	After separating the special cases $s=1$ and $s=2$ in the sums $\Sigma_1$ and $\Sigma_2$, one obtains
	\begin{align*}
		\Sigma_1+\Sigma_2=&-\sum_{k=j+2}^n\mu^{k-j+1}\partial_k\mu^1-\sum_{k=j+2}^n\mu^{k-j}\partial_k\mu^2\notag\\
		&+\sum_{s=3}^{n-j+1}\mu^{h}c^s_{hk}\partial_{j+s-1}\mu^k-\sum_{s=3}^{n-j+1}\mu^{h}c^k_{j+s-1,h}\partial_{k}\mu^s,
	\end{align*}
	where the terms in the third sum that correspond to $k\geq3$ cancel with the fourth sum, by a relabelling of indices. Then
	\begin{align*}
		\Sigma_1+\Sigma_2=&-\sum_{k=j+2}^n\mu^{k-j+1}\partial_k\mu^1-\sum_{k=j+2}^n\mu^{k-j}\partial_k\mu^2+\sum_{s=3}^{n-j+1}\sum_{k=1}^2\mu^{h}c^s_{hk}\partial_{j+s-1}\mu^k,
	\end{align*}
	where the first (second) sum cancels with the part of the third sum that corresponds to $k=1$ ($k=2$), yielding $\Sigma_1+\Sigma_2=0$ and thus $\pa_jV=0$. Since $j\in\{2,\dots,n\}$ was arbitrary, we have proved
	\begin{equation}
		\pa_jV=0,\qquad j\in\{2,\dots,n\}.
		\label{eq:higherVvanish-oneblock}
	\end{equation}
	
	\textsc{Step 2: $\pa_j\mu^1=0$, $j\ge3$.} We argue by descending induction on $j$, starting from $j=n$. By inspecting the equation \eqref{eq:EQ2gen} for $i=1$ and $j=n-1$, every term vanishes except the one containing $\mu^2\pa_n\mu^1$, and two other terms that simplify with each other. Hence
	\[
	0=\mn^2\,\pa_n\mn^1.
	\]
	The assumption that $\mu^2\neq0$ implies $\pa_n\mn^1=0$. Let us now fix $j\in\{3,\dots,n-1\}$ and inductively assume that $\pa_\ell\mn^1=0$ for all $\ell\ge j+1$. We consider \eqref{eq:EQ2gen} for $i=1$, replacing $j$ by $j-1$. Up to a simplification, the surviving terms are those involving $\pa_j\mn^1$, as the terms with $\pa_\ell\mn^1$ for $\ell\ge j+1$ vanish by the induction hypothesis. Therefore
	\[
	0=\mn^2\,\pa_j\mn^1.
	\]
	As above, this implies $\pa_j\mn^1=0$.
	
	\textsc{Step 3: $\pa_1V=0$ and $\pa_2\mu^1=0$.} Let us sum the diagonal equations obtained by considering \eqref{eq:EQ2gen} for $j=i$, where the sum runs over $i=2,\dots,n$. Because of \eqref{eq:higherVvanish-oneblock} and \emph{Step 2}, the only surviving pieces are those involving $\pa_2\mu^1$ and $\pa_1V$, and terms that simplify with each other. A straightforward calculation then yields
	\begin{equation}
		n\,\mu^2\,\pa_2\mn^1=(n-1)\,\pa_1V.
		\label{eq:identity-I}
	\end{equation}
	We now use the equation \eqref{eq:EQ1gen} for $(i,j)=(1,2)$. Again, using \eqref{eq:higherVvanish-oneblock}, one sees that all second derivatives of $V$ disappear except those involving $\pa_1V$, and the equation reduces to
	\begin{equation}
		(\pa_2\mn^1)(\pa_1V)=0.
		\label{eq:identity-II}
	\end{equation}
	Suppose by contradiction that $\pa_1V\neq0$ at some point. Then \eqref{eq:identity-II} forces $\pa_2\mn^1=0$ at that point, and \eqref{eq:identity-I} gives
	\[
	0=n\,\mu^2\,\pa_2\mn^1=(n-1)\,\pa_1V,
	\]
	which is impossible. Hence $\pa_1V=0$ everywhere. Similarly, $\partial_2\mu^1=0$.
\end{proof}

\subsection{The general regular case}
We now return to the general block decomposition. The previous summation argument repeats independently inside each block because all mixed structure constants vanish. According to the notation where $r$ is the number of Jordan blocks, with sizes $m_1,\dots,m_r$, the system  \eqref{eq:EQ1gen}--\eqref{eq:EQ2gen} reads
\begin{subequations}\label{EQ12blocks}
	\begin{align}\label{EQ1blocks}
		&\mu^{t(\tau)}\bigg(c^{s(\sigma)}_{t(\tau)i(\alpha)}\partial_{s(\sigma)}\partial_{j(\beta)}V-c^{s(\sigma)}_{t(\tau)j(\beta)}\partial_{s(\sigma)}\partial_{i(\alpha)}V\bigg)\\
		&-\bigg(c^{t(\tau)}_{j(\beta)s(\sigma)}\partial_{i(\alpha)}\mu^{s(\sigma)}-c^{t(\tau)}_{i(\alpha)s(\sigma)}\partial_{j(\beta)}\mu^{s(\sigma)}\bigg)\partial_{t(\tau)}V=0,\notag\\
		\label{EQ2blocks}
		&\mu^{h(\tau)}\bigg(c^{i(\alpha)}_{h(\tau)k(\gamma)}\partial_{j(\beta)}\mu^{k(\gamma)}+c^{i(\alpha)}_{j(\beta)k(\gamma)}\partial_{h(\tau)}\mu^{k(\gamma)}-c^{k(\gamma)}_{j(\beta)h(\tau)}\partial_{k(\gamma)}\mu^{i(\alpha)}\\
		&-c^{i(\alpha)}_{j(\beta)s(\sigma)}c^{s(\sigma)}_{h(\tau)k(\gamma)}e^{m(\nu)}\partial_{m(\nu)}\mu^{k(\gamma)}\bigg)+e^{i(\alpha)}\partial_{j(\beta)}V-\delta^{i(\alpha)}_{j(\beta)}e(V)=0,\notag
	\end{align}
\end{subequations}
for all $\alpha,\beta\in\{1,\dots,r\}$ and $i\in\{1,\dots,m_\alpha\}$, $j\in\{1,\dots,m_\beta\}$.

\begin{proposition}\label{non_ex_gen}
	Let $V$ and $\{\{\mu^{i(\alpha)}\}_{i\in\{1,\dots,m_\alpha\}}\}_{\alpha\in\{1,\dots,r\}}$ solve the system \eqref{EQ12blocks}. Then, for $\alpha,\beta\in\{1,\dots,r\}$, one has
	\begin{itemize}
		\item[(i)] for $m_\alpha\geq2,$
        \begin{equation}\label{v2vanish}
			\partial_{j(\alpha)}V=0,\quad j\in\{2,\dots,m_\alpha\},
		\end{equation}
		\item[(ii)] for $m_\beta\geq2,$
		\begin{equation*}
			\partial_{j(\beta)}\mu^{1(\alpha)}=0,\quad j\in\{2,\dots,m_\beta\},
		\end{equation*}
        \item[(iii)] the classical Gibbons--Tsarev system is recovered
		\begin{align}
			\label{GTblocks1}
			\partial_{1(\beta)}V&=(\mu^{1(\beta)}-\mu^{1(\alpha)})\partial_{1(\beta)}\mu^{1(\alpha)},\qquad\,\, \beta\neq\alpha,\\
			\label{GTblocks2}
            \partial_{1(\alpha)}\partial_{1(\beta)}V&=2\frac{(\partial_{1(\alpha)}V)(\partial_{1(\beta)}V)}{(\mu^{1(\beta)}-\mu^{1(\alpha)})^2},\qquad\qquad\quad \beta\neq\alpha,
		\end{align}
		for the leading variables $\{u^{1(\alpha)}\}_{\alpha\in\{1,\dots,r\}},$
		\item[(iv)] for $m_\alpha\geq2,$
		\begin{equation*}
			\partial_{1(\alpha)}V=0.
		\end{equation*}
	\end{itemize}
\end{proposition}
\begin{proof}
	The proof naturally splits into multiple parts.
	
	\textsc{Step 1: $\partial_{j(\alpha)}V=0$ for $j\in\{2,\dots,m_\alpha\}$.} Let $\alpha\in\{1,\dots,r\}$ such that $m_\alpha\geq2$ and let $j\geq2$. By taking the sum, over $s\in\{1,\dots,m_\alpha-j+1\}$, of the quantities appearing in \eqref{EQ2blocks} for $\beta=\alpha$ and $i=s$, replacing $j$ by $j+s-1$, we have
	\begin{align}
		0=&\overset{m_\alpha-j+1}{\underset{s=1}{\sum}}\bigg\{\mu^{h(\tau)}\bigg(c^{s(\alpha)}_{h(\tau)k(\gamma)}\partial_{(j+s-1)(\alpha)}\mu^{k(\gamma)}+\underbrace{c^{s(\alpha)}_{(j+s-1)(\alpha)k(\gamma)}}_{\delta^\alpha_\gamma c^{1(\alpha)}_{j(\alpha)k(\alpha)}=0}\partial_{h(\tau)}\mu^{k(\gamma)}-c^{k(\gamma)}_{(j+s-1)(\alpha)h(\tau)}\partial_{k(\gamma)}\mu^{s(\alpha)}\notag\\
		&\qquad\qquad\quad-\underbrace{c^{s(\alpha)}_{(j+s-1)(\alpha)t(\sigma)}}_{\delta^\alpha_\sigma c^{1(\alpha)}_{j(\alpha)t(\alpha)}=0}c^{t(\sigma)}_{h(\tau)k(\gamma)}e^{m(\nu)}\partial_{m(\nu)}\mu^{k(\gamma)}\bigg)+\underbrace{e^{s(\alpha)}\partial_{(j+s-1)(\alpha)}V}_{\delta^s_1\partial_{j(\alpha)}V}-\underbrace{\delta^{s(\alpha)}_{(j+s-1)(\alpha)}}_{\delta^j_1=0}e(V)\bigg\}\notag\\
		=&\mu^{h(\alpha)}\bigg(c^{1(\alpha)}_{h(\alpha)k(\alpha)}\partial_{j(\alpha)}\mu^{k(\alpha)}-c^{k(\alpha)}_{j(\alpha)h(\alpha)}\partial_{k(\alpha)}\mu^{1(\alpha)}\bigg)+\partial_{j(\alpha)}V\notag\\
		&+\mu^{h(\alpha)}\bigg(c^{2(\alpha)}_{h(\alpha)k(\alpha)}\partial_{(j+1)(\alpha)}\mu^{k(\alpha)}-c^{k(\alpha)}_{(j+1)(\alpha)h(\alpha)}\partial_{k(\alpha)}\mu^{2(\alpha)}\bigg)\notag\\
		&+\overset{m_\alpha-j+1}{\underset{s=3}{\sum}}\mu^{h(\alpha)}\bigg(\underbrace{c^{s(\alpha)}_{h(\alpha)k(\alpha)}\partial_{(j+s-1)(\alpha)}\mu^{k(\alpha)}}_{\text{split the sum over $k$, with respect to $k=1,2$ and $k\geq3$}}-c^{k(\alpha)}_{(j+s-1)(\alpha)h(\alpha)}\partial_{k(\alpha)}\mu^{s(\alpha)}\bigg)\notag\\
		=&-\overset{m_\alpha}{\underset{k=j+2}{\sum}}\mu^{h(\alpha)}c^{k(\alpha)}_{j(\alpha)h(\alpha)}\partial_{k(\alpha)}\mu^{1(\alpha)}+\partial_{j(\alpha)}V-\overset{m_\alpha}{\underset{k=j+2}{\sum}}\mu^{h(\alpha)}c^{k(\alpha)}_{(j+1)(\alpha)h(\alpha)}\partial_{k(\alpha)}\mu^{2(\alpha)}\notag\\
		&+\overset{m_\alpha-j+1}{\underset{s=3}{\sum}}\mu^{s(\alpha)}\partial_{(j+s-1)(\alpha)}\mu^{1(\alpha)}+\overset{m_\alpha-j+1}{\underset{s=3}{\sum}}\mu^{(s-1)(\alpha)}\partial_{(j+s-1)(\alpha)}\mu^{2(\alpha)}\notag\\			
		&+\overset{m_\alpha-j+1}{\underset{s=3}{\sum}}\overset{m_\alpha}{\underset{k=3}{\sum}}\mu^{h(\alpha)}c^{s(\alpha)}_{h(\alpha)k(\alpha)}\partial_{(j+s-1)(\alpha)}\mu^{k(\alpha)}-\overset{m_\alpha-j+1}{\underset{s=3}{\sum}}\mu^{h(\alpha)}c^{k(\alpha)}_{(j+s-1)(\alpha)h(\alpha)}\partial_{k(\alpha)}\mu^{s(\alpha)},\notag
	\end{align}
	where the first and third sums cancel, as do the second and the fourth sums, and the fifth and sixth sums. This yields $\partial_{j(\alpha)}V=0$, proving (i).
	
	\textsc{Step 2: $\partial_{j(\alpha)}\mu^{1(\alpha)}=0$ for $j\in\{3,\dots,m_\alpha\}$.} In order to prove that
	\begin{equation}
		\label{LemmaGT1_eq1}
		\partial_{j(\alpha)}\mu^{1(\alpha)}=0,\qquad j\geq3,
	\end{equation}
	let us consider \eqref{EQ2blocks} with $\beta=\alpha$, $i=1$ and $j=m_\alpha-1$. This gives
	\begin{align}
		0=&\mu^{h(\alpha)}\bigg(c^{1(\alpha)}_{h(\alpha)k(\alpha)}\partial_{(m_\alpha-1)(\alpha)}\mu^{k(\alpha)}-c^{k(\alpha)}_{(m_\alpha-1)(\alpha)h(\alpha)}\partial_{k(\alpha)}\mu^{1(\alpha)}\bigg)+e^{1(\alpha)}\partial_{(m_\alpha-1)(\alpha)}V\notag\\
		\overset{\eqref{v2vanish}}{=}&-\mu^{2(\alpha)}\partial_{m_\alpha(\alpha)}\mu^{1(\alpha)},\notag
	\end{align}
	yielding $\partial_{m_\alpha(\alpha)}\mu^{1(\alpha)}=0$. By induction, let us fix $l\in\{3,\dots,m_\alpha-1\}$ and assume that
	\begin{equation}\label{LemmaGT1_eq1_ind}
		\partial_{m(\alpha)}\mu^{1(\alpha)}=0,\qquad m\in\{l+1,\dots,m_\alpha\},
	\end{equation}
	and show $\partial_{l(\alpha)}\mu^{1(\alpha)}=0$. Considering \eqref{EQ2blocks} with $\beta=\alpha$, $i=1$ and $j=l-1$ gives
	\begin{align}
		0=&\mu^{h(\alpha)}\bigg(c^{1(\alpha)}_{h(\alpha)k(\alpha)}\partial_{(l-1)(\alpha)}\mu^{k(\alpha)}-c^{k(\alpha)}_{(l-1)(\alpha)h(\alpha)}\partial_{k(\alpha)}\mu^{1(\alpha)}\bigg)+e^{1(\alpha)}\partial_{(l-1)(\alpha)}V\notag\\
		\overset{\eqref{v2vanish}}{=}&-\overset{m_\alpha-l+2}{\underset{h=2}{\sum}}\mu^{h(\alpha)}\partial_{(l+h-2)(\alpha)}\mu^{1(\alpha)}\overset{\eqref{LemmaGT1_eq1_ind}}{=}-\mu^{2(\alpha)}\partial_{l(\alpha)}\mu^{1(\alpha)},\notag
	\end{align}
	yielding $\partial_{l(\alpha)}\mu^{1(\alpha)}=0$.		
	
	\textsc{Step 3: $\partial_{j(\beta)}\mu^{1(\alpha)}=0$ for $j\in\{2,\dots,m_\beta\}$, $\beta\neq\alpha$.} In order to prove that
	\begin{equation}
		\label{LemmaGT1_eq2}
		\partial_{j(\beta)}\mu^{1(\alpha)}=0,\qquad j\geq2,\quad\beta\neq\alpha,
	\end{equation}
	let us consider \eqref{EQ2blocks} with $\beta\neq\alpha$, $i=1$ and $j=m_\beta$. This gives
	\begin{align}
		0=&(\mu^{1(\alpha)}-\mu^{1(\beta)})\partial_{m_\beta(\beta)}\mu^{1(\alpha)}+\partial_{m_\beta(\beta)}V\overset{\eqref{v2vanish}}{=}(\mu^{1(\alpha)}-\mu^{1(\beta)})\partial_{m_\beta(\beta)}\mu^{1(\alpha)},\notag
	\end{align}
	yielding $\partial_{m_\beta(\beta)}\mu^{1(\alpha)}=0$. By induction, let us fix $l\in\{2,\dots,m_\beta-1\}$ and assume that
	\begin{equation}\label{LemmaGT1_eq2_ind}
		\partial_{m(\beta)}\mu^{1(\alpha)}=0,\qquad m\in\{l+1,\dots,m_\beta\},
	\end{equation}
	and show $\partial_{l(\beta)}\mu^{1(\alpha)}=0$. Considering \eqref{EQ2blocks} with $\beta\neq\alpha$, $i=1$ and $j=l$ gives
	\begin{align}
		0=&(\mu^{1(\alpha)}-\mu^{1(\beta)})\partial_{l(\beta)}\mu^{1(\alpha)}-\overset{m_\beta}{\underset{h=2}{\sum}}\mu^{h(\beta)}c^{k(\beta)}_{l(\beta)h(\beta)}\partial_{k(\beta)}\mu^{1(\alpha)}+\partial_{l(\beta)}V\notag\\
		\overset{\eqref{v2vanish}}{=}&(\mu^{1(\alpha)}-\mu^{1(\beta)})\partial_{l(\beta)}\mu^{1(\alpha)}-\overset{m_\beta}{\underset{h=2}{\sum}}\mu^{h(\beta)}c^{k(\beta)}_{l(\beta)h(\beta)}\partial_{k(\beta)}\mu^{1(\alpha)}
	\end{align}
	where the sum vanishes by \eqref{LemmaGT1_eq2_ind}, yielding $\partial_{l(\beta)}\mu^{1(\alpha)}=0$.
	
	\textsc{Step 4: \eqref{GTblocks1} and \eqref{GTblocks2} hold.} Let $\beta\neq\alpha$. Considering $i=j=1$ in \eqref{EQ2blocks}  gives
	\begin{align}
		0=&\mu^{h(\tau)}\bigg(c^{1(\alpha)}_{h(\tau)k(\gamma)}\partial_{1(\beta)}\mu^{k(\gamma)}-c^{k(\gamma)}_{1(\beta)h(\tau)}\partial_{k(\gamma)}\mu^{1(\alpha)}\bigg)+\partial_{1(\beta)}V\notag\\
		=&\mu^{h(\alpha)}c^{1(\alpha)}_{h(\alpha)k(\alpha)}\partial_{1(\beta)}\mu^{k(\alpha)}-\mu^{h(\beta)}c^{k(\beta)}_{1(\beta)h(\beta)}\partial_{k(\beta)}\mu^{1(\alpha)}+\partial_{1(\beta)}V\notag
	\end{align}
	which, by virtue of \eqref{LemmaGT1_eq1}--\eqref{LemmaGT1_eq2}, gives \eqref{GTblocks1}. Moreover, considering $i=j=1$ in \eqref{EQ1blocks} gives
	\begin{align}
		0=&\mu^{t(\alpha)}c^{s(\alpha)}_{t(\alpha)1(\alpha)}\partial_{s(\alpha)}\partial_{1(\beta)}V-\mu^{t(\beta)}c^{s(\beta)}_{t(\beta)1(\beta)}\partial_{s(\beta)}\partial_{1(\alpha)}V\notag\\
		&-c^{t(\beta)}_{1(\beta)s(\beta)}\partial_{1(\alpha)}\mu^{s(\beta)}\partial_{t(\beta)}V+c^{t(\alpha)}_{1(\alpha)s(\alpha)}\partial_{1(\beta)}\mu^{s(\alpha)}\partial_{t(\alpha)}V\notag\\
		\overset{\eqref{v2vanish}}{=}&(\mu^{1(\alpha)}-\mu^{1(\beta)})\partial_{1(\alpha)}\partial_{1(\beta)}V-\partial_{1(\alpha)}\mu^{1(\beta)}\partial_{1(\beta)}V+\partial_{1(\beta)}\mu^{1(\alpha)}\partial_{1(\alpha)}V
	\end{align}
	which, by means of \eqref{GTblocks1}, gives \eqref{GTblocks2}. This proves (iii).
	
	\textsc{Step 5: $\partial_{1(\alpha)}V=0$, $\partial_{2(\alpha)}\mu^{1(\alpha)}=0$ for $m_\alpha\geq2$.} Let $m_\alpha\geq2$.  By taking the sum, over $i\in\{2,\dots,m_\alpha\}$, of the quantities appearing in \eqref{EQ2blocks} for $\beta=\alpha$ and $j=i$, we have
	\begin{align}
		0=&\overset{m_\alpha}{\underset{i=2}{\sum}}\bigg\{
		\mu^{h(\alpha)}c^{i(\alpha)}_{h(\alpha)k(\alpha)}\partial_{i(\alpha)}\mu^{k(\alpha)}+\mu^{h(\tau)}\partial_{h(\tau)}\mu^{1(\alpha)}-\mu^{h(\alpha)}c^{k(\alpha)}_{i(\alpha)h(\alpha)}\partial_{k(\alpha)}\mu^{i(\alpha)}\notag\\
		&\qquad-\mu^{1(\alpha)}e^{m(\nu)}\partial_{m(\nu)}\mu^{1(\alpha)}-e(V)
		\bigg\}\notag\\
		=&\underset{\beta\neq\alpha}{\sum}\overset{m_\alpha}{\underset{i=2}{\sum}}\bigg\{
		\underbrace{\mu^{h(\beta)}\partial_{h(\beta)}\mu^{1(\alpha)}-\mu^{1(\alpha)}\partial_{1(\beta)}\mu^{1(\alpha)}-\partial_{1(\beta)}V}_{\overset{\eqref{LemmaGT1_eq2}}{=}(\mu^{1(\beta)}-\mu^{1(\alpha)})\partial_{1(\beta)}\mu^{1(\alpha)}-\partial_{1(\beta)}V\overset{\eqref{GTblocks1}}{=}0}
		\bigg\}\notag\\
		&+\overset{m_\alpha}{\underset{i=2}{\sum}}\bigg\{
		\mu^{h(\alpha)}c^{i(\alpha)}_{h(\alpha)k(\alpha)}\partial_{i(\alpha)}\mu^{k(\alpha)}+\mu^{h(\alpha)}\partial_{h(\alpha)}\mu^{1(\alpha)}-\mu^{h(\alpha)}c^{k(\alpha)}_{i(\alpha)h(\alpha)}\partial_{k(\alpha)}\mu^{i(\alpha)}\notag\\
		&\qquad\quad\,-\mu^{1(\alpha)}\partial_{1(\alpha)}\mu^{1(\alpha)}-\partial_{1(\alpha)}V
		\bigg\}\notag\\
		=&\overset{m_\alpha}{\underset{i=2}{\sum}}\overset{m_\alpha}{\underset{h=2}{\sum}}\mu^{h(\alpha)}\bigg\{
		c^{i(\alpha)}_{h(\alpha)k(\alpha)}\partial_{i(\alpha)}\mu^{k(\alpha)}+\partial_{h(\alpha)}\mu^{1(\alpha)}-c^{k(\alpha)}_{i(\alpha)h(\alpha)}\partial_{k(\alpha)}\mu^{i(\alpha)}
		\bigg\}-(m_\alpha-1)\partial_{1(\alpha)}V\notag\\
		\overset{\eqref{LemmaGT1_eq1}}{=}&(m_\alpha-1)\mu^{2(\alpha)}\partial_{2(\alpha)}\mu^{1(\alpha)}+\underbrace{
			\overset{m_\alpha}{\underset{h=2}{\sum}}\mu^{h(\alpha)}\overset{m_\alpha}{\underset{i=2}{\sum}}\bigg\{
			c^{i(\alpha)}_{h(\alpha)1(\alpha)}\partial_{i(\alpha)}\mu^{1(\alpha)}-c^{1(\alpha)}_{i(\alpha)h(\alpha)}\partial_{1(\alpha)}\mu^{i(\alpha)}
			\bigg\}
		}_{\overset{\eqref{LemmaGT1_eq1}}{=}\mu^{2(\alpha)}\partial_{2(\alpha)}\mu^{1(\alpha)}}\notag\\
		&+\overset{m_\alpha}{\underset{h=2}{\sum}}\mu^{h(\alpha)}\underbrace{
			\overset{m_\alpha}{\underset{i=2}{\sum}}\overset{m_\alpha}{\underset{k=2}{\sum}}\bigg\{
			c^{i(\alpha)}_{h(\alpha)k(\alpha)}\partial_{i(\alpha)}\mu^{k(\alpha)}-c^{k(\alpha)}_{i(\alpha)h(\alpha)}\partial_{k(\alpha)}\mu^{i(\alpha)}
			\bigg\}
		}_{0}-(m_\alpha-1)\partial_{1(\alpha)}V\notag\\
		=&m_\alpha\mu^{2(\alpha)}\partial_{2(\alpha)}\mu^{1(\alpha)}-(m_\alpha-1)\partial_{1(\alpha)}V\notag
	\end{align}
	yielding
	\begin{equation}\label{v1vanish_proof1}
		m_\alpha\mu^{2(\alpha)}\partial_{2(\alpha)}\mu^{1(\alpha)}=(m_\alpha-1)\partial_{1(\alpha)}V.
	\end{equation}
	Considering \eqref{EQ1blocks} with $\beta=\alpha$, $i=1$ and $j=2$, and taking into account \eqref{v2vanish}, gives
	\begin{equation}\label{v1vanish_proof2}
		\partial_{2(\alpha)}\mu^{1(\alpha)}\,\partial_{1(\alpha)}V=0.
	\end{equation}
	From \eqref{v1vanish_proof1} and \eqref{v1vanish_proof2}, it follows that $\partial_{1(\alpha)}V=0$ and $\partial_{2(\alpha)}\mu^{1(\alpha)}=0$. This proves (iv) and, together with \emph{Step 2} and \emph{Step 3}, (ii).
\end{proof}

\section{Reductions of Pavlov's hydrodynamic chain}
In this section, we study reductions
   \[C^i=C^i(r^1,...,r^n),\qquad r^i_t=c^i_{jk}(\mu^j-Ve^j)r^k_x\]
of \emph{Pavlov's hydrodynamic chain} 
\begin{equation}\label{PHC}
C^{\alpha-1}_t=C^{\alpha}_x-C^0C^{\alpha-1}_x,\qquad \alpha\in\mathbb{Z}.
\end{equation}
Diagonal reductions (corresponding to a semisimple product in our setting) have been studied by Pavlov in \cite{pavlov}.  He proved that the characteristic velocities $\mu^i$ of the diagonal reductions are given by
\[\mu^i=f^i(r^i)-\sum_{k=1}^n\psi^k(r^k),\]
where $f^i(r^i),\psi_i(r^i),\,i=1,...,n$ are arbitrary functions. The starting observation is that, upon identifying the coordinates $(r^1,...,r^n)$ with the  canonical coordinates of a semisimple product, the functions $(f_1(r^1),\dots,f_n(r^n))$ define an eventual identity
 $\mu$ and the function $V=\sum_{k=1}^n\psi^k(r^k)$ is the general solution of the equation $d\cdot d_L V=0$ in the case $L=\mu\circ$. In other words, in the semisimple case reductions have the form
\begin{equation}\label{redPHC}
r_t=(\mu-V e)\circ r_x=(L-V I)r_x,
\end{equation}
where the vector field  $\mu$ and the function $V$ are solutions of the system
\begin{subequations}\label{GTmod}
\begin{align}
\mathcal{L}_{\mu}c^i_{hj}-c^i_{js}c^s_{ht}[e,\mu]^t&=0, \label{gtm1}\\
(c^s_{ih}\partial_s\partial_j V-c^s_{jh}\partial_s\partial_i V)\mu^h&=
(c^h_{sj}\partial_i\mu^s-c^h_{si}\partial_j\mu^s)\partial_hV\label{gtm2}.
\end{align}
\end{subequations} 
Thanks to this interpretation, we can extend this result to the general case, i.e.  without assuming  semisimplicity. We will need the following two Propositions as auxiliary tools.

\begin{proposition}\label{pro:ddomega}
    
Let $M$ be a smooth manifold, let
\[
\omega \in \Omega^k(M;TM)=\Gamma(\Lambda^kT^*M\otimes TM),
\]
and let $\Omega^\bullet(M)$ denote the de Rham algebra of differential forms on $M$. Define the insertion operator
\[
\iota_\omega:\Omega^p(M)\longrightarrow \Omega^{p+k-1}(M)
\]
as the unique derivation of degree $k-1$ characterised on decomposable tensors by
\[
\iota_{\alpha\otimes X}\eta=\alpha\wedge \iota_X\eta,
\qquad
\alpha\in \Omega^k(M),\ X\in \mathfrak X(M),\ \eta\in \Omega^\bullet(M).
\]
Equivalently, if $p\ge 1$, then
\[
(\iota_\omega\eta)(X_1,\dots,X_{p+k-1})
=
\sum_{\sigma\in \operatorname{Sh}(k,p-1)}
\operatorname{sgn}(\sigma)\,
\eta\bigl(
\omega(X_{\sigma(1)},\dots,X_{\sigma(k)}),
X_{\sigma(k+1)},\dots,X_{\sigma(k+p-1)}
\bigr),
\]
and $\iota_\omega f=0$ for every $f\in C^\infty(M)$.

Now define
\[
d_\omega
:=
[\iota_\omega,d]_{\mathrm{gr}}
=
\iota_\omega\circ d-(-1)^{k-1}d\circ \iota_\omega.
\]
Then the following hold.
\begin{enumerate}
\item
$d_\omega$ is a derivation of degree $k$, i.e.
\[
d_\omega:\Omega^p(M)\longrightarrow \Omega^{p+k}(M).
\]

\item
If $\omega=\alpha\otimes X$ with $\alpha\in \Omega^k(M)$ and $X\in \mathfrak X(M)$, then for every $\eta\in \Omega^\bullet(M)$,
\[
d_\omega\eta
=
\alpha\wedge \mathcal L_X\eta
+
(-1)^k\,d\alpha\wedge \iota_X\eta.
\]

\item
The operators $d$ and $d_\omega$ always graded-commute:
\[
[d,d_\omega]_{\mathrm{gr}}
=
d\circ d_\omega-(-1)^k d_\omega\circ d
=
0.
\]
Equivalently,
\[
d\circ d_\omega = (-1)^k\, d_\omega\circ d.
\]
Hence:
\begin{itemize}
\item if $k$ is even, then $d$ and $d_\omega$ commute;
\item if $k$ is odd, then $d$ and $d_\omega$ anticommute.
\end{itemize}

\item
Assume now that $k=1$. Under the canonical identification
\[
\Omega^1(M;TM)\cong \Gamma(T^*M\otimes TM)\cong \Gamma(\operatorname{End}(TM)),
\]
let
\[
L:=\omega^\flat \in \Gamma(\operatorname{End}(TM))
\qquad\text{be defined by}\qquad
LX=\omega(X).
\]
Define the usual insertion operator
\[
\iota_L:\Omega^p(M)\longrightarrow \Omega^p(M)
\]
by
\[
(\iota_L\eta)(X_1,\dots,X_p)
=
\sum_{j=1}^p \eta(X_1,\dots,LX_j,\dots,X_p),
\]
and set
\begin{equation}\label{eq:dLnew}
d_L:=\iota_L\circ d-d\circ \iota_L.
\end{equation}
Then
\[
\iota_\omega=\iota_L
\qquad\text{and therefore}\qquad
d_\omega=d_L.
\]

\item
For every $q\ge 0$, every $\eta\in \Omega^q(M)$, and every $X_0,\dots,X_q\in \mathfrak X(M)$,
\[
\begin{aligned}
(d_L\eta)(X_0,\dots,X_q)
&=
\sum_{i=0}^q (-1)^i (LX_i)\bigl(\eta(X_0,\dots,\widehat{X_i},\dots,X_q)\bigr) \\
&\quad +
\sum_{0\leq i<j\leq q}(-1)^{i+j}\,
\eta([X_i,X_j]_L, X_0,\dots,\widehat{X_i},\dots,\widehat{X_j},\dots,X_q),
\end{aligned}
\]
where
\[
[X_i,X_j]_L:=[LX_i, X_j]+[X_i, LX_j]-L[X_i,X_j],
\]
that is the operator $d_L$ defined in \eqref{eq:dLnew} coincides with the operator $d_L$ defined in \eqref{eq:dLoriginal}.
\end{enumerate}
\end{proposition}

\begin{proof}
Since $\iota_\omega$ is a derivation of degree $k-1$ and $d$ is a derivation of degree $1$, their graded commutator
\[
d_\omega=[\iota_\omega,d]_{\mathrm{gr}}
\]
is a derivation of degree $k$. This proves (1).

For (2), if $\omega=\alpha\otimes X$, then $\iota_\omega=\alpha\wedge \iota_X$, so for every $\eta\in \Omega^\bullet(M)$,
\[
\begin{aligned}
d_\omega\eta
&=
\iota_\omega(d\eta)-(-1)^{k-1}d(\iota_\omega\eta) \\
&=
\alpha\wedge \iota_X(d\eta)
-
(-1)^{k-1}d(\alpha\wedge \iota_X\eta) \\
&=
\alpha\wedge \iota_X(d\eta)
-
(-1)^{k-1}\Bigl(d\alpha\wedge \iota_X\eta+(-1)^k\alpha\wedge d(\iota_X\eta)\Bigr) \\
&=
\alpha\wedge\bigl(\iota_X d+d\iota_X\bigr)\eta
+
(-1)^k\,d\alpha\wedge \iota_X\eta \\
&=
\alpha\wedge \mathcal L_X\eta
+
(-1)^k\,d\alpha\wedge \iota_X\eta,
\end{aligned}
\]
where we used Cartan's formula $\mathcal L_X=\iota_Xd+d\iota_X$.

For (3), using $d^2=0$, we compute
\[
d\,d_\omega
=
d\iota_\omega d-(-1)^{k-1}d^2\iota_\omega
=
d\iota_\omega d,
\]
while
\[
d_\omega d
=
\iota_\omega d^2-(-1)^{k-1}d\iota_\omega d
=
(-1)^k\, d\iota_\omega d.
\]
Therefore
\[
d\,d_\omega = (-1)^k\, d_\omega d,
\]
which is exactly
\[
[d,d_\omega]_{\mathrm{gr}}=0.
\]
The parity statement follows immediately.

For (4), let $\eta\in \Omega^p(M)$ and $X_1,\dots,X_p\in \mathfrak X(M)$. Since here $k=1$, the shuffle formula gives
\[
(\iota_\omega\eta)(X_1,\dots,X_p)
=
\sum_{j=1}^p (-1)^{j-1}
\eta\bigl(\omega(X_j),X_1,\dots,\widehat{X_j},\dots,X_p\bigr).
\]
Because $\eta$ is alternating, moving $\omega(X_j)$ from the first slot to the $j$-th slot produces another factor $(-1)^{j-1}$, hence
\[
(\iota_\omega\eta)(X_1,\dots,X_p)
=
\sum_{j=1}^p
\eta\bigl(X_1,\dots,\omega(X_j),\dots,X_p\bigr).
\]
By definition of $L=\omega^\flat$, one has $LX_j=\omega(X_j)$, so
\[
(\iota_\omega\eta)(X_1,\dots,X_p)
=
\sum_{j=1}^p
\eta\bigl(X_1,\dots,LX_j,\dots,X_p\bigr)
=
(\iota_L\eta)(X_1,\dots,X_p).
\]
Thus $\iota_\omega=\iota_L$. Since here $k=1$,
\[
d_\omega
=
\iota_\omega d-d\iota_\omega
=
\iota_L d-d\iota_L
=
d_L.
\]

It remains to prove (5). Fix $q\ge 0$, $\eta\in \Omega^q(M)$, and $X_0,\dots,X_q\in \mathfrak X(M)$. By definition,
\[
d_L\eta=\iota_L(d\eta)-d(\iota_L\eta).
\]

For $r\neq i$, write
\[
\eta_i^{(r)}
:=
\eta(X_0,\dots,\widehat{X_i},\dots,LX_r,\dots,X_q),
\]
that is, $\eta_i^{(r)}$ is obtained from
$\eta(X_0,\dots,\widehat{X_i},\dots,X_q)$ by replacing the argument $X_r$ by $LX_r$.
Also, for $0\leq i<j\leq q$, define
\[
\eta_{ij}(Y)
:=
\eta(Y,X_0,\dots,\widehat{X_i},\dots,\widehat{X_j},\dots,X_q),
\]
and for $r\neq i,j$ define
\[
\eta_{ij}^{(r)}(Y)
:=
\eta(Y,X_0,\dots,\widehat{X_i},\dots,\widehat{X_j},\dots,LX_r,\dots,X_q).
\]

We first expand $\iota_L d\eta$. Since
\[
(\iota_L d\eta)(X_0,\dots,X_q)
=
\sum_{r=0}^q (d\eta)(X_0,\dots,LX_r,\dots,X_q),
\]
the usual de Rham formula gives
\[
\begin{aligned}
(\iota_L d\eta)(X_0,\dots,X_q)
&=
\sum_{r=0}^q (-1)^r\,LX_r\bigl(\eta(X_0,\dots,\widehat{X_r},\dots,X_q)\bigr) \\
&\quad +
\sum_{i=0}^q (-1)^i \sum_{r\neq i} X_i(\eta_i^{(r)}) \\
&\quad +
\sum_{0\leq i<j\leq q}(-1)^{i+j}\,
\eta_{ij}([LX_i,X_j]+[X_i,LX_j]) \\
&\quad +
\sum_{0\leq i<j\leq q}(-1)^{i+j}
\sum_{r\neq i,j}\eta_{ij}^{(r)}([X_i,X_j]).
\end{aligned}
\]

Next we expand $d(\iota_L\eta)$. Since
\[
(\iota_L\eta)(X_0,\dots,\widehat{X_i},\dots,X_q)
=
\sum_{r\neq i}\eta_i^{(r)},
\]
another application of the de Rham formula yields
\[
\begin{aligned}
(d\iota_L\eta)(X_0,\dots,X_q)
&=
\sum_{i=0}^q (-1)^i X_i\left(\sum_{r\neq i}\eta_i^{(r)}\right) \\
&\quad +
\sum_{0\leq i<j\leq q}(-1)^{i+j}
(\iota_L\eta)([X_i,X_j],X_0,\dots,\widehat{X_i},\dots,\widehat{X_j},\dots,X_q) \\
&=
\sum_{i=0}^q (-1)^i \sum_{r\neq i} X_i(\eta_i^{(r)}) \\
&\quad +
\sum_{0\leq i<j\leq q}(-1)^{i+j}
\sum_{r\neq i,j}\eta_{ij}^{(r)}([X_i,X_j]) \\
&\quad +
\sum_{0\leq i<j\leq q}(-1)^{i+j}\eta_{ij}(L[X_i,X_j]).
\end{aligned}
\]

Subtracting the two expansions, all mixed derivative terms
\[
\sum_{i=0}^q (-1)^i \sum_{r\neq i} X_i(\eta_i^{(r)})
\]
cancel, and all mixed bracket terms
\[
\sum_{0\leq i<j\leq q}(-1)^{i+j}
\sum_{r\neq i,j}\eta_{ij}^{(r)}([X_i,X_j])
\]
also cancel. Therefore
\[
\begin{aligned}
(d_L\eta)(X_0,\dots,X_q)
&=
\sum_{r=0}^q (-1)^r\,LX_r\bigl(\eta(X_0,\dots,\widehat{X_r},\dots,X_q)\bigr) \\
&\quad +
\sum_{0\leq i<j\leq q}(-1)^{i+j}\,
\eta_{ij}\bigl([LX_i,X_j]+[X_i,LX_j]-L[X_i,X_j]\bigr).
\end{aligned}
\]
Since
\[
\eta_{ij}(Y)
=
\eta(Y,X_0,\dots,\widehat{X_i},\dots,\widehat{X_j},\dots,X_q),
\]
the previous expression is exactly \eqref{eq:dLoriginal}.

\end{proof}

\begin{proposition}
Let $L$ be an invertible tensor field of type $(1,1)$ with vanishing Nijenhuis torsion.  Let
\[
\tau_L \colon \Omega^k(M)\to \Omega^k(M)
\]
be the natural action of $L$ on $k$-forms, defined by
\[
(\tau_L\omega)(X_1,\dots,X_k):=\omega(LX_1,\dots,LX_k).
\]
Then the following hold. 
\begin{enumerate}
\item $L^{-1}$ has also vanishing Nijenhuis torsion.
\item For every $\alpha\in \Omega(M)$,
\[
d\,d_L\alpha
=
-\,\tau_L\bigl(d\,d_{L^{-1}}(\tau_{L^{-1}}\alpha)\bigr).
\]
Equivalently,
\[
d\,d_L
=
-\,\tau_L\circ d\,d_{L^{-1}}\circ \tau_{L^{-1}}.
\]
Hence, for every $\alpha\in \Omega(M)$,
\[
d\,d_L\alpha=0
\quad\Longleftrightarrow\quad
d\,d_{L^{-1}}(\tau_{L^{-1}}\alpha)=0.
\]

\item If $V\in C^\infty(M)=\Omega^0(M)$, then
\[
d\,d_LV=0
\quad\Longleftrightarrow\quad
d\,d_{L^{-1}}V=0.
\]
\end{enumerate}
\end{proposition}

\begin{proof} 

(1).  Set $A := L^{-1}$. Recall that the Nijenhuis torsion of a $(1,1)$-tensor $T$ is defined by
\[
N_T(X,Y)
:=
[T X, T Y]
- T\bigl([T X,Y] + [X,T Y]\bigr)
+ T^2[X,Y]
\]
for vector fields $X,Y \in \mathfrak{X}(M)$.

We want to prove that $N_A=0$, assuming that $N_L=0$.  Let $X,Y \in \mathfrak{X}(M)$. Since $L$ is invertible, every pair of vector fields can be written in the form
\[
U = L X,
\qquad
V = L Y
\]
for unique  $X,Y$.  Thus it is enough to show that
\[
N_A(LX,LY)=0 \quad \forall X,Y\in \mathfrak{X}(M).
\]

\noindent Using $A(LX)=X$ and $A(LY)=Y$, we compute
\[
\begin{aligned}
N_A(LX,LY)
&=
[A(LX),A(LY)]
- A\bigl([A(LX),LY] + [LX,A(LY)]\bigr)
+ A^2[LX,LY] \\
&=
[X,Y]
- A\bigl([X,LY] + [LX,Y]\bigr)
+ A^2[LX,LY].
\end{aligned}
\]
Applying $L^2$ to both sides and using $A=L^{-1}$, we obtain
\[
\begin{aligned}
L^2 N_A(LX,LY)
&=
L^2[X,Y]
- L\bigl([X,LY] + [LX,Y]\bigr)
+ [LX,LY] \\
&=
[LX,LY]
- L\bigl([LX,Y] + [X,LY]\bigr)
+ L^2[X,Y] \\
&=
N_L(X,Y).
\end{aligned}
\]
By hypothesis,  $N_L=0$, hence
\[
L^2 N_A(LX,LY)=0.
\]
but since  $L^2$ is invertible,  we get the claim. 

(2). First we show that,  since $N_L=0$,  one has
\[
d_L\circ \tau_L=\tau_L\circ d.
\]
Indeed, for $\omega\in \Omega^k(M)$ and $X_0,\dots,X_k\in \mathfrak X(M)$,  we have
\begin{align*}
(d_L(\tau_L\omega))(X_0,\dots,X_k)
&=
\sum_{i=0}^k (-1)^i (LX_i)\bigl(\tau_L\omega)(X_0,\dots,\widehat{X_i},\dots,X_k)\bigr)
\\
&\quad
+\sum_{0\le i<j\le k}(-1)^{i+j}
\,(\tau_L\omega)\bigl([X_i,X_j]_L,X_0,\dots,\widehat{X_i},\dots,\widehat{X_j},\dots,X_k\bigr)\\
&=
\sum_{i=0}^k (-1)^i (LX_i)\bigl(\omega(LX_0,\dots,\widehat{LX_i},\dots,LX_k)\bigr)
\\
&\quad
+\sum_{0\le i<j\le k}(-1)^{i+j}
\,\omega\bigl(L[X_i,X_j]_L,LX_0,\dots,\widehat{LX_i},\dots,\widehat{LX_j},\dots,LX_k\bigr).
\end{align*}
Since
\[
N_L(X_i,X_j)=[LX_i,LX_j]-L[X_i,X_j]_L,
\]
the assumption $N_L=0$ gives
\[
L[X_i,X_j]_L=[LX_i,LX_j].
\]
Therefore substituting we get:
\begin{align*}
(d_L(\tau_L\omega))(X_0,\dots,X_k)
&=
\sum_{i=0}^k (-1)^i (LX_i)\bigl(\omega(LX_0,\dots,\widehat{LX_i},\dots,LX_k)\bigr)
\\
&\quad
+\sum_{0\le i<j\le k}(-1)^{i+j}
\,\omega\bigl([LX_i,LX_j],LX_0,\dots,\widehat{LX_i},\dots,\widehat{LX_j},\dots,LX_k\bigr)
\\
&=
(d\omega)(LX_0,\dots,LX_k)
\\
&=
(\tau_L(d\omega))(X_0,\dots,X_k).
\end{align*}
Thus
\[
d_L\circ \tau_L=\tau_L\circ d.
\]
Since $L$ is invertible, $\tau_L$ is invertible with inverse $\tau_{L^{-1}}$, and hence
\[
d_L=\tau_L\circ d\circ \tau_{L^{-1}}.
\]
Applying the same identity to $L^{-1}$, whose Nijenhuis torsion also vanishes, gives
\[
d_{L^{-1}}=\tau_{L^{-1}}\circ d\circ \tau_L
\qquad\text{and}\qquad
d\circ \tau_L=\tau_L\circ d_{L^{-1}}.
\]

\noindent Now let $\alpha\in \Omega(M)$. Using the previous identities and the anticommutation relation (see Proposition \ref{pro:ddomega})
\[
d\,d_{L^{-1}}+d_{L^{-1}}d=0,
\]
we compute
\begin{align*}
d\,d_L\alpha
&=
d\bigl(\tau_L(d(\tau_{L^{-1}}\alpha))\bigr)
\\
&=
\tau_L\bigl(d_{L^{-1}}d(\tau_{L^{-1}}\alpha)\bigr)
\\
&=
-\,\tau_L\bigl(d\,d_{L^{-1}}(\tau_{L^{-1}}\alpha)\bigr).
\end{align*}
This proves the identity in part (2), and since $\tau_L$ is invertible it follows that
\[
d\,d_L\alpha=0
\quad\Longleftrightarrow\quad
d\,d_{L^{-1}}(\tau_{L^{-1}}\alpha)=0.
\]

(3).  Let $V\in C^\infty(M)$. Since $V$ is a $0$-form,
\[
\tau_{L^{-1}}V=V.
\]
Therefore part (2) gives
\[
d\,d_LV=-\,\tau_L(d\,d_{L^{-1}}V).
\]
Now $d\,d_{L^{-1}}V$ is a $2$-form, and $\tau_L$ is invertible on $\Omega^2(M)$. Hence
\[
d\,d_LV=0
\quad\Longleftrightarrow\quad
d\,d_{L^{-1}}V=0.
\]
This proves part (3).
\end{proof}

\begin{proposition}\label{Th_red_PHC}
All reductions
     \[C^i=C^i(r^1,...,r^n),\qquad r^i_t=W^i_{j}(r)r^j_x,\]
of Pavlov's hydrodynamic chain
\begin{equation*}
C^{\alpha-1}_t=C^{\alpha}_x-C^0C^{\alpha-1}_x,\qquad \alpha\in\mathbb{Z},
\end{equation*}
are obtained from a $(1,1)$ tensor field with vanishing Nijenhuis torsion $L$ 
 and from a solution $V$ of the equation $d\cdot  d_L V=0$ by the following procedure:
 \begin{itemize}
 \item the $(1,1)$ tensor field $W$ is defined as $W=L-VI$,
 \item $C^0=V$ and the remaining functions $C^k$ are determined by the following recursive relations
\begin{eqnarray}
\label{rec1}
dC^{\alpha+1}&=&d_LC^{\alpha},\qquad  \alpha=0,1,2,...\\
\label{rec2}
dC^{-\alpha-1}&=&d_{L^{-1}}C^{\alpha},\qquad  \alpha=0,1,2,....
\end{eqnarray}
 \end{itemize}
\end{proposition}

\begin{proof}
First of all we observe that the functions $C^\alpha$ are (locally) well-defined. This follows
 immediately from the  equation $d\cdot d_L V=0$ (i.e. \eqref{gtm2}), by the previous propositions and by the properties  of the complexes defined by $(d,d_L)$ and by $(d,d_{L^{-1}})$. 
   
Putting $C^{\alpha}=C^{\alpha}(r^1,...,r^n)$ in the Pavlov hydrodynamic chain  \eqref{PHC} and assuming
\begin{equation}\label{redPHC2}
r_t=(L-V I)r_x,
\end{equation}
 we get the system (the two terms with $C_0$ cancel out):
\begin{equation}\label{rec}
(L^h_k\partial_h C^{\alpha-1}-\partial_kC^{\alpha})r^k_x=0,
\end{equation}
that  is satisfied due to \eqref{rec1} and \eqref{rec2}. 

Conversely, we observe  that it is not restrictive to write $W$ in the form $W=L-C^0I$ obtaining  the system
\[\partial_kC^{\alpha}=L^h_k\partial_h C^{\alpha-1}.\]
It is immediate to check that Darboux's compatibility of this system coincides with the vanishing of the Nijenhuis torsion of $L$. Indeed
\[
\partial_j\partial_kC^{\alpha}=(\partial_jL^h_k)\partial_hC^{\alpha-1}+L^h_k\partial_h \partial_jC^{\alpha-1}=(\partial_jL^h_k)L^i_h\partial_iC^{\alpha-2}+L^h_k\partial_h(L^i_j)\partial_iC^{\alpha-2}+L^h_kL^i_j\partial_h\partial_iC^{\alpha-2},
\]
and thus
\[\partial_j\partial_kC^{\alpha}-\partial_k\partial_jC^{\alpha}=[(\partial_jL^h_k-\partial_kL^h_j)L^i_h+L^h_k\partial_hL^i_j-L^h_j\partial_hL^i_k]\partial_iC^{\alpha-2}=N^i_{jk}\partial_iC^{\alpha-2}.\]
Finally, applying the differential $d$ to both sides we get
  $d\cdot d_L C^{\alpha-1}=0$.
\end{proof}

From Proposition \ref{NiEvId} we have the following Corollary.
\begin{corollary}
In the special case where $L=\mu\circ$ comes from a regular $F$-manifold structure
 $(\circ,e)$ the reductions are defined by solutions 
\[\mu^1(r^1,\dots,r^n),\dots,\mu^n(r^1,\dots,r^n), \qquad
V(r^1,\dots,r^n),
\]
of  the system \eqref{GTmod}.
\end{corollary}

\subsection{Some examples}
In all the examples $\mu=E=\sum_i r^i\frac{\partial}{\partial r^i}$.
\subsubsection{$2\times 2$ Jordan block}
\begin{equation*}
L=\begin{bmatrix}
r_{1} & 0\cr
r_{2} & r_{1} \cr
\end{bmatrix}.
\end{equation*}
\[V=F_1(r_1)r_2+F_2(r_1).\]

\subsubsection{$3\times 3$ Jordan block}
\begin{equation*}
L=\begin{bmatrix}
r_{1} & 0 & 0\cr
r_{2} & r_{1} &  0\cr
r_{3} & r_{2} & r_{1}
\end{bmatrix}.
\end{equation*}
\[V=F_1( r_1)r_3+\frac{1}{2}F'_1(r_1)r_2^{2}+F_2(r_1) r_{2}+F_3(r_1) 
.\]

\subsubsection{$2\times 2+1\times 1$ Jordan blocks}
\begin{equation}\label{Lblock}
L=\begin{bmatrix}
r_{1} & 0 & 0\cr
r_{2} & r_{1} &  0\cr
0 & 0 & r_{3}
\end{bmatrix}.
\end{equation}
\[V=F_1(r_1)r_2+F_2(r_1)+F_3(r_3).\]

\subsubsection{$4\times 4$ Jordan block}
\begin{equation*}
L=\begin{bmatrix}
r_{1} & 0 & 0 & 0\cr
r_{2} & r_{1} & 0 & 0\cr
r_3 & r_2 & r_1 & 0\cr
r_{4} & r_3 & r_{2} & r_{1}
\end{bmatrix}.
\end{equation*}
\[V=\left( F'_1( r_1) r_2+F_2( r_1)\right)r_3+ F_1( r_1) r_4+\frac{1}{6}F''_1(r_1)r_2^{3}+\frac{1}{2}F_2'(r_1)r_2^{2}+F_3(r_1)r_2+F_4(r_1).\]

\subsubsection{$3\times 3+1\times 1$ Jordan blocks}
\begin{equation*}
L=\begin{bmatrix}
r_{1} & 0 & 0 & 0\cr
r_{2} & r_{1} & 0 & 0\cr
r_3 & r_2 & r_1 & 0\cr
0 & 0 & 0 & r_{4}
\end{bmatrix}.
\end{equation*}
\[V=F_1( r_1)r_{3}+\frac{1}{2}F'_1(r_1)r_2^{2}+F_2(r_1)r_2+F_3( r_1)+F_4(r_4).\]

\subsubsection{$2\times 2+2\times 2$ Jordan blocks}
\begin{equation*}
L=\begin{bmatrix}
r_{1} & 0 & 0 & 0\cr
r_{2} & r_{1} & 0 & 0\cr
0 & 0 & r_3 & 0\cr
0 & 0 & r_{4} & r_{3}
\end{bmatrix}.
\end{equation*}
\[V=F_1(r_1)r_2+F_2(r_1)+F_3(r_3)r_4+F_4(r_3).\]

\subsubsection{$2\times 2+1\times 1+1\times 1$ Jordan blocks}
\begin{equation*}
L=\begin{bmatrix}
r_{1} & 0 & 0 & 0\cr
r_{2} & r_{1} & 0 & 0\cr
0 & 0 & r_3 & 0\cr
0 & 0 & 0 & r_{4}
\end{bmatrix}.
\end{equation*}
\[V=F_1(r_1)r_2+F_2(r_1)+F_3(r_3)+F_4(r_4).\]

\subsubsection{$5\times 5$ Jordan block}
\begin{equation*}
L=\begin{bmatrix}
r_{1} & 0 & 0 & 0 & 0\cr
r_{2} & r_{1} & 0 & 0 & 0\cr
r_3 & r_2 & r_1 & 0 & 0\cr
r_{4} & r_3 & r_{2} & r_{1} &0 \cr
r_{5} & r_4 & r_{3} & r_{2} &r_1
\end{bmatrix}.
\end{equation*}
\begin{eqnarray*}
V&=&(F'_1( r_1)r_2+F_2( r_1))r_4+F_1( r_1)r_5+\frac{1}{2}F_1'(r_1)r_3^2+\left(\frac{1}{2}F_1''(r_1)r_2^2+F_2'(r_1)r_2+F_3(r_1)\right)r_3\\&+&\frac{1}{24}F_1'''(r_1)r_2^4+\frac{1}{6}F_2''(r_1)r_2^3+\frac{1}{2}F_3'(r_1)r_2^2+F_4(r_1)r_2+F_5( r_1).
\end{eqnarray*}

\subsubsection{$4\times 4+1\times 1$ Jordan blocks}
\begin{equation*}
L=\begin{bmatrix}
r_{1} & 0 & 0 & 0 & 0\cr
r_{2} & r_{1} & 0 & 0 & 0\cr
r_3 & r_2 & r_1 & 0 & 0\cr
r_4 & r_3 & r_2 & r_1 & 0\cr
0 & 0 & 0 & 0 & r_{5}
\end{bmatrix}.
\end{equation*}
\[V=\left( F'_1( r_1) r_2+F_2( r_1)\right)r_3+ F_1( r_1)r_4+\frac{1}{6}F''_1(r_1)r_2^{3}+\frac{1}{2}F_2'(r_1)r_2^{2}+F_3(r_1)r_2+F_4(r_1)+F_5(r_5).\] 

\subsubsection{$3\times 3+2\times 2$ Jordan blocks}
\begin{equation*}
L=\begin{bmatrix}
r_{1} & 0 & 0 & 0 & 0\cr
r_{2} & r_{1} & 0 & 0 &0\cr
r_{3} & r_{2} & r_1 & 0 &0\cr
0 & 0 & 0 & r_4 & 0\cr
0 & 0 & 0 & r_{5} & r_{4}
\end{bmatrix}.
\end{equation*}
\[V=F_1( r_1)r_{{3}}+\frac{1}{2}F'_1(r_1)r_2^{2}+F_2(r_1) r_2+F_3(r_1)+F_4(r_4)r_5+F_5(r_4)
.\]

\subsubsection{$3\times 3+1\times 1+1\times 1$ Jordan blocks}
\begin{equation*}
L=\begin{bmatrix}
r_{1} & 0 & 0 & 0 & 0\cr
r_{2} & r_{1} & 0 & 0 &0\cr
r_{3} & r_{2} & r_1 & 0 &0\cr
0 & 0 & 0 & r_4 & 0\cr
0 & 0 & 0 & 0 & r_{5}
\end{bmatrix}.
\end{equation*}
\[V=F_1( r_1)r_{3}+\frac{1}{2}F'_1(r_1)r_2^{2}+F_2(r_1) r_2+F_3(r_1)+F_4(r_4)+F_5(r_5)
.\]

\subsubsection{$2\times 2+2\times 2+1\times 1$ Jordan blocks}
\begin{equation*}
L=\begin{bmatrix}
r_{1} & 0 & 0 & 0& 0\cr
r_{2} & r_{1} & 0 & 0 & 0\cr
0 & 0 & r_3 & 0 & 0 \cr
0 & 0 & r_4 & r_3 & 0 \cr
0 & 0 & 0 & 0 & r_{5}
\end{bmatrix}.
\end{equation*}
\[V=F_1(r_1)r_2+F_2(r_1)+F_3(r_3)r_4+F_4(r_3)+F_5(r_5).\]

\subsubsection{$2\times 2+1\times 1+1\times 1+1\times1$ Jordan blocks}
\begin{equation*}
L=\begin{bmatrix}
r_{1} & 0 & 0 & 0 & 0\cr
r_{2} & r_{1} & 0 & 0 & 0\cr
0 & 0 & r_3 & 0 & 0\cr
0 & 0 & 0 & r_{4} & 0\cr
0 & 0 & 0 & 0 & r_5
\end{bmatrix}.
\end{equation*}
\[V=F_1(r_1)r_2+F_2(r_1)+F_3(r_3)+F_4(r_4)+F_5(r_5).\]

In the regular case, choosing as eventual identity the Euler vector, the general solution
 of the equation \eqref{gtm2} was obtained in \cite{LPVG1} combining  two results:
 \begin{itemize}
 \item     The general solution of the equation $d\cdot d_L V=0$, for $L$ having a single Jordan block of size $n$, is given by
    \begin{equation}
        V = F_n(r^1) + \sum_{s>0} \dfrac{1}{s!} \sum_{k_1=2}^n \cdots \sum_{k_s=2}^{n} r^{k_1}\cdots r^{k_s}\left(\frac{d}{d r^1}\right)^{(s-1)}F_{n+s-\sum_{j=1}^s k_j}(r^1),
        \label{eq:gensol}
    \end{equation}
    where $F_1, \dots, F_n$ are arbitrary functions of $r^1$, and $F_{n+s-\sum_{j=1}^{s}k_j} = 0$ whenever ${n+s-\sum_{j=1}^{s}k_j \leq 0}$.
 \item A splitting Lemma reducing the solution in the case of multiple  Jordan blocks to the single Jordan block case  (this lemma has been also proved in \cite{BKM} with different methods).
 \end{itemize}
 
\section{Conclusions} 
In this paper we study finite-component regular reductions of dKP and Pavlov's hydrodynamic chain defined by systems of hydrodynamic type associated with a regular $F$-manifold structure $(\circ,e,E)$. They are described  by the system
\begin{subequations}
\begin{align}
(\mathcal{L}_{\mu}c^i_{hj}-c^i_{js}c^s_{ht}[e,\mu]^t)\mu^h&=e(V)\delta^i_j-e^i\partial_jV, \\
(c^s_{ih}\partial_s\partial_j V-c^s_{jh}\partial_s\partial_i V)\mu^h&=
(c^h_{sj}\partial_i\mu^s-c^h_{si}\partial_j\mu^s)\partial_hV, 
\end{align}
\end{subequations}
and by the system
\begin{subequations}
\begin{align}
\mathcal{L}_{\mu}c^i_{hj}-c^i_{js}c^s_{ht}[e,\mu]^t&=0,\\
(c^s_{ih}\partial_s\partial_j V-c^s_{jh}\partial_s\partial_i V)\mu^h&=
(c^h_{sj}\partial_i\mu^s-c^h_{si}\partial_j\mu^s)\partial_hV,
\end{align}
\end{subequations} 
respectively. 
In the first case, the tensor field defining the system has the form $W=\mu\circ$ while in the second case $W=(\mu-Ve)\circ$. We showed that in the case of dKP non-trivial reductions are possible only when the product $\circ$ is semisimple while in the case of Pavlov's hydrodynamic chain non-trivial reductions exist for any number of Jordan blocks and for Jordan blocks of arbitrary size.


\begin{thebibliography}{99}
\bibitem{ALimrn} 
A. Arsie and P. Lorenzoni,  \emph{F-manifolds with eventual identities, bidifferential calculus and twisted Lenard-Magri chains},
 International Mathematics Research Notices, Vol. 2013, No. 17, 3931--3976.
 
\bibitem{BKM} A. Bolsinov, A.Y. Konyaev and V.S. Matveev, \emph{Nijenhuis geometry IV: conservation laws, symmetries and integration of certain non-diagonalisable systems of hydrodynamic type in quadratures}, Nonlinearity 37 (10) 105003 (2024). 

\bibitem{DH}
L. David and C. Hertling, \emph{Regular F-manifolds: initial conditions and Frobenius metrics}, Ann. Sc. Norm. Super. Pisa Cl. Sci. (5)
Vol. XVII (2017), 1121--1152.

\bibitem{DS}
L. David and I. A. B. Strachan, \emph{Dubrovin's duality for $F$-manifolds with eventual identities}, Adv. Math. Volume 226, Issue 5, 20 March 2011, Pages 4031--4060.

\bibitem{HM} C. Hertling and Yu. I. Manin, \emph{Weak Frobenius manifolds}, International Mathematics Research Notices, Vol. 1999, No. 6, 277--286 (1999).

\bibitem{FK}
E. V. Ferapontov, K. R. Khusnutdinova, \emph{On the Integrability of (2+1)-Dimensional Quasilinear Systems}, Commun. Math. Phys. 248, 187–206 (2004). 

\bibitem{FN} A. Fr\"{o}licher, A. Nijenhuis,  \emph{Theory of vector-valued differential forms}, Proc. Ned. Acad. Wetensch. Ser. A 59 (1956), 338--359.

\bibitem{GLR} J. Gibbons, P. Lorenzoni, A. Raimondo,
\emph{Hamiltonian structures of reductions of the Benney system},
  Commun. Math. Phys. 287, 291–322 (2009).

\bibitem{GT1} J. Gibbons, S.P. Tsarev,
\emph{Reductions of the Benney equations},
 Phys. Lett. A {\bf 211} (1996), no. 1, 19--24.

\bibitem{GT2}  J. Gibbons, S.P. Tsarev, 
\emph{Conformal maps and reductions of the Benney equations},
  Phys. Lett. A  258  (1999),  no. 4-6, 263--271.

\bibitem{LM} P. Lorenzoni and F. Magri, \emph{A cohomological construction of integrable hierarchies of hydrodynamic type}, International Mathematics Research Notices, Volume 2005, No. 34, 2005, 2087--2100. 
  
\bibitem{LPR}
P. Lorenzoni, M. Pedroni, A. Raimondo, \emph{F-manifolds and integrable systems
 of hydrodynamic type}, Archivum Mathematicum 
47 (2011), 163--180.  

\bibitem{LP23} P. Lorenzoni and S. Perletti, \emph{Integrable systems, Frölicher–Nijenhuis bicomplexes and Lauricella bi-flat structures}, Nonlinearity 36(12), pp. 6925–6990 (2023).

\bibitem{LPVG1} P. Lorenzoni,  S.  Perletti,  and K. van Gemst,  \emph{Integrable Hierarchies and F-Manifolds with Compatible Connection}, Commun.  Math.  Phys.  Volume 406,  article number 91,  (2025). 

\bibitem{LPVG2} P. Lorenzoni,  S.  Perletti,  and K. van Gemst,  \emph{The generalised hodograph method for non-diagonalisable integrable systems of hydrodynamic type}, Nonlinearity 39 (2026).

\bibitem{LPVG3} P. Lorenzoni, S.  Perletti and K. van Gemst \emph{Semi-Hamiltonian Properties of a Class of Non-diagonalisable Systems of Hydrodynamic Type}, Commun. Math. Phys. 407, 46 (2026).

\bibitem{M}
Yu. I. Manin, \emph{$F$-manifolds with flat structure and Dubrovin's duality}, Adv. Math. Volume 198, issue 1, (2005), 5--26.

\bibitem{pavlov}
M. V. Pavlov,  \emph{Integrable hydrodynamic chains}, J. Math. Phys. 44, 4134--4156 (2003).

\bibitem{PS}
S. Perletti and I.A.B. Strachan, \emph{Regular $F$-manifolds with eventual identities}, J. Phys. A: Math. Theor. Volume 57, 475201 (2024).

\bibitem{ts91} S. P. Tsarev,
\emph{The geometry of Hamiltonian systems of hydrodynamic type. The generalised hodograph transform}, USSR Izv. 37 (1991) 397--419.

\end{thebibliography}
\end{document}